\documentclass[12pt,a4paper]{article}

\usepackage{amsmath,amssymb}
\usepackage{graphicx}
\usepackage[nosort]{cite}

\usepackage{tikz}
\usetikzlibrary{decorations.markings}
\usepackage{epsfig}  
\usepackage{enumerate}
\usepackage{relsize}
\newcommand{\mathsym}[1]{{}}

\usepackage{amsthm,amsbsy,color,multicol}
\usepackage{array,calc}
\usepackage{rotating}
\usepackage{a4wide,bm}
\usepackage[a4paper,text={450pt,650pt},centering]{geometry}
\usepackage{setspace}

\let\oldbfseries=\bfseries
\let\oldmdseries=\mdseries
\let\oldnormalfont=\normalfont
\renewcommand{\bfseries}{\oldbfseries\boldmath}
\renewcommand{\mdseries}{\oldmdseries\unboldmath}
\renewcommand{\normalfont}{\oldnormalfont\unboldmath}

\allowdisplaybreaks[3]

\numberwithin{equation}{section}

\usepackage[font=small,labelfont=bf,width=0.85\textwidth]{caption}

\ifx\hypersetup\sadfkjashdfkxja\newcommand\hypersetup[1]{}\fi

\hypersetup{plainpages=false}
\hypersetup{pdfpagemode=UseNone}
\hypersetup{bookmarksnumbered=true}
\hypersetup{pdfstartview=FitH}
\hypersetup{colorlinks=false}
\hypersetup{citebordercolor={.5 1 .5}}
\hypersetup{urlbordercolor={.5 1 1}}
\hypersetup{linkbordercolor={1 .7 .7}}

\DeclareMathSymbol{\Gamma}{\mathalpha}{letters}{"00}
\DeclareMathSymbol{\Delta}{\mathalpha}{letters}{"01}
\DeclareMathSymbol{\Theta}{\mathalpha}{letters}{"02}
\DeclareMathSymbol{\Lambda}{\mathalpha}{letters}{"03}
\DeclareMathSymbol{\Xi}{\mathalpha}{letters}{"04}
\DeclareMathSymbol{\Pi}{\mathalpha}{letters}{"05}
\DeclareMathSymbol{\Sigma}{\mathalpha}{letters}{"06}
\DeclareMathSymbol{\Upsilon}{\mathalpha}{letters}{"07}
\DeclareMathSymbol{\Phi}{\mathalpha}{letters}{"08}
\DeclareMathSymbol{\Psi}{\mathalpha}{letters}{"09}
\DeclareMathSymbol{\Omega}{\mathalpha}{letters}{"0A}

\newcommand{\dd}{\mathrm{d}}
\newcommand{\ii}{\mathrm{i}}

\newcommand*\widebar[1]{%
  \hbox{%
    \vbox{%
      \hrule height 0.5pt 
      \kern0.25ex
      \hbox{%
        \kern-0.3em
        \ensuremath{#1}%
        \kern-0.1em
      }%
    }%
  }%
}

\ifx\genfrac\sdflkaj\else\fi

\newcommand{\ket}[1]{\left|#1\right\rangle}      
\newcommand{\bra}[1]{\left\langle #1\right|}     

\newcommand{\alg}[1]{\mathfrak{#1}}

\newcommand{\beq}{\begin{equation}}
\newcommand{\eeq}{\end{equation}}

\def\[{\begin{equation}}
\def\]{\end{equation}}
\def\<{\begin{eqnarray}}
\def\>{\end{eqnarray}}

\newtheorem{lemma}{Lemma} 
\newtheorem{remark}{Remark}

\makeatletter
\def\mr@ignsp#1 {\ifx\:#1\@empty\else #1\expandafter\mr@ignsp\fi}%
\newcommand{\multiref}[1]{\begingroup
\xdef\mr@no@sparg{\expandafter\mr@ignsp#1 \: }%
\def\mr@comma{}%
\@for\mr@refs:=\mr@no@sparg\do{\mr@comma\def\mr@comma{,}\ref{\mr@refs}}%
\endgroup}
\makeatother

\newcommand{\hypref}[2]{\ifx\href\asklfhas #2\else\href{#1}{#2}\fi}
\newcommand{\Secref}[1]{Section~\multiref{#1}}

\newcommand{\Appref}[1]{Appendix~\multiref{#1}}

\newcommand{\Figref}[1]{Figure~\multiref{#1}}

\renewcommand{\eqref}[1]{(\multiref{#1})}

\makeatletter
\newlength{\apb@width}
\newcommand{\autoparbox}[2][c]{\settowidth{\apb@width}{#2}\parbox[#1]{\apb@width}{#2}}

\makeatother

\ifx\href\asklfhas\newcommand{\href}[2]{#2}\fi

\begin{document}

\renewcommand{\thefootnote}{\fnsymbol{footnote}}
\thispagestyle{empty}
\begin{flushright}\footnotesize
ITP-UU-12/47 \\
SPIN-12/44
\end{flushright}
\vspace{1cm}

\begin{center}%
{\Large\bfseries%
\hypersetup{pdftitle={Scalar product of Bethe vectors from functional equations}}%
Scalar product of Bethe vectors \\ from functional equations%
\par} \vspace{2cm}%

\textsc{W. Galleas}\vspace{5mm}%
\hypersetup{pdfauthor={Wellington Galleas}}%

\textit{Institute for Theoretical Physics and Spinoza Institute, \\ Utrecht University, Leuvenlaan 4,
3584 CE Utrecht, \\ The Netherlands}\vspace{3mm}%

\verb+w.galleas@uu.nl+ %

\par\vspace{3cm}

\textbf{Abstract}\vspace{7mm}

\begin{minipage}{12.7cm}
In this work the scalar product of Bethe vectors for the six-vertex model
is studied by means of functional equations. The scalar products are
shown to obey a system of functional equations originated from the Yang-Baxter
algebra and its solution is given as a multiple contour integral.

\hypersetup{pdfkeywords={Six-vertex model, functional equations, scalar product}}%
\hypersetup{pdfsubject={}}%
\end{minipage}
\vskip 2cm
{\small December 2012}

\end{center}

\newpage
\renewcommand{\thefootnote}{\arabic{footnote}}
\setcounter{footnote}{0}

\tableofcontents

\section{Introduction}
\label{sec:intro}

Some important models of quantum field theory and condensed matter physics
exhibit remarkable properties granting them the gift of integrability.
Among prominent examples we have the quantum non-linear Schr\"odinger model
\cite{Sklyanin_1979}, the Heisenberg spin chain \cite{Heisenberg_1928} and 
the Hubbard model \cite{Hubbard_1963}.
The realm of those models are of quantum nature and, although a rigorous and
unambiguous definition of integrability at quantum level is still lacking 
\cite{Caux_2011, Weigert_1992, Faddeev_2004, Marmo_2009}, these models exhibit 
a set of enhanced symmetries which can be explored in order to compute physical
quantities exactly \cite{Korepin_book, Hubb_book}. Among those quantities we mention
the energy spectrum and some correlation functions, and in this way integrability
offers non-perturbative access to the behaviour of strongly interacting systems.
These achievements are in large part due to the advent of the Bethe ansatz \cite{Bethe_1931}
and its algebraic formulation \cite{Sk_Faddeev_1979, Takh_Faddeev_1979} characterising the model
wave function in terms of its momentum. This algebraic formulation has
been successfully applied to the majority of known integrable systems 
\cite{Tarasov_1988, Ramos_1998, Galleas_2004, Melo_2009}, yielding the model
eigenvectors in terms of creation operators acting on a suitable reference state.
Moreover, as remarked in \cite{McCoy_2001} the study of Bethe's wave function also provided
important insights leading to the concept of commuting transfer matrices \cite{Baxter_1971}.

On the other hand, the exact solution of a model whose transfer matrix belongs to a 
commutative family is also intimately connected with functional equations methods \cite{Baxter_1972}.
As far as spectral properties are concerned, we have available a variety of functional methods
\cite{Baxter_1972, Reshet_1987, Stroganov_1979, Galleas_2008} yielding transfer matrices
eigenvalues which do not address the problem of explicitly constructing the respective eigenvectors. 
In fact, the evaluation of most physical quantities do not require the eigenvectors
themselves but quantities which can be derived from them. For instance, the calculation of 
correlation functions would then require the evaluation of eigenvectors scalar products
in addition to the expected value of operators \cite{Korepin_book}.

Within the framework of the algebraic Bethe ansatz, the calculation of scalar product of Bethe
vectors was inaugurated by Korepin \cite{Korepin82} largely influenced by Gaudin's hypothesis
on the norms of the non-linear Schr\"odinger model wave function \cite{gaudin_book}.
The method of \cite{Korepin82} has also been described in \cite{Korepin_book} where the scalar
product of Bethe vectors for the six-vertex model is given as the vacuum mean value of a determinant
whose entries are expressed in terms of quantum fields. In \cite{Korepin_book, Maillet_1999} that scalar product
is also given as a summation over a product of determinants with scalar entries.

The purpose of this paper is to demonstrate that scalar product of Bethe vectors
can also be computed by means of functional equations. The method we shall employ here
resembles the one described in the series of works \cite{Galleas10, Galleas11,Galleas_2011,Galleas_2012}
where the Yang-Baxter algebra has played the major role in deriving functional equations describing
the partition functions of vertex and SOS models with domain wall boundaries.
The origin of the functional equations describing scalar products is also the
Yang-Baxter algebra and its solution turns out to be given by a multiple contour
integral. 

This paper is organised as follows. In \Secref{sec:bethe} we give a brief description
of Bethe vectors for the six-vertex model. In \Secref{sec:FN} we illustrate how the 
Yang-Baxter algebra can be explored in order to derive functional equations describing
scalar products. The solution of our functional equations is obtained in \Secref{sec:scalar}
and concluding remarks are discussed in \Secref{sec:conclusion}. Technical details and
proofs are presented through the Appendices A to E.

\section{Bethe vectors for the six-vertex model}
\label{sec:bethe}

Vertex models in two dimensions constitute one of the corner stones of the
theory of exactly solvable models of statistical mechanics \cite{Baxter_book}.
In particular, the study of the eigenvectors of the six-vertex model \cite{Lieb_1967}
by means of the Bethe ansatz \cite{Bethe_1931} provided insightful information
which paved the way to establish the connection between two-dimensional vertex
models and one-dimensional spin chains \cite{McCoy_1968, Sutherland_1970, Baxter_1971}.
This study also received a large impulse with the advent of the Quantum Inverse 
Scattering Method (QISM) which unveiled the algebraic foundation supporting
the construction of Bethe vectors \cite{Sk_Faddeev_1979, Takh_Faddeev_1979}.  

\paragraph{Bethe vectors.} In the framework of the QISM, the eigenvectors of the six-vertex
model are built up from the action of creation operators on a pseudo-vacuum state. 
More precisely, an eigenvector $\ket{\psi}$ is of the form
\[
\label{psi}
\ket{\psi} = B(\lambda_1^B) \dots B(\lambda_n^B) \ket{0}
\]
where $B(\lambda)$, together with three more generators $A(\lambda)$, $C(\lambda)$
and $D(\lambda)$, satisfy a certain set of algebraic relations.
In its turn, the vector $\ket{0}$ is the $\alg{sl}_2$ highest weight vector
while the parameters $\lambda_j^B \in \mathbb{C}$ are required to satisfy certain 
constraints. 

\medskip

\paragraph{Yang-Baxter algebra.} The algebra satisfied by the generators $A(\lambda)$, 
$B(\lambda)$, $C(\lambda)$ and $D(\lambda)$ is commonly referred to as Yang-Baxter algebra
and it reads
\[
\label{yb_alg}
\mathcal{R}_{12} (\lambda_1 - \lambda
_2) \mathcal{T}_1 (\lambda_1) \mathcal{T}_2 (\lambda_2) = \mathcal{T}_2 (\lambda_2) \mathcal{T}_1 (\lambda_1) \mathcal{R}_{12} (\lambda_1 - \lambda_2) \; .
\]
The relation (\ref{yb_alg}) is defined in $\mbox{End} (\mathbb{V}_1 \otimes \mathbb{V}_2 )$ with $\mathbb{V}_j \cong \mathbb{C}^2$.
In their turn, $\mathcal{T}_1 = \mathcal{T} \otimes \mbox{id}$ and $\mathcal{T}_2 = \mbox{id} \otimes \mathcal{T}$ with
\<
\label{abcd}
\mathcal{T} (\lambda) = \left( \begin{matrix}
A(\lambda) & B(\lambda) \cr
C(\lambda) & D(\lambda) \end{matrix} \right) \; .
\>
The matrix $\mathcal{R}_{a b} \in \mbox{End} (\mathbb{V}_a \otimes \mathbb{V}_b)$ encodes the algebra structure constants
associated with the six-vertex model and it is explicitly given by 
\<
\label{rmat}
\mathcal{R} = \left( \begin{matrix}
a & 0 & 0 & 0 \cr
0 & b & c & 0 \cr
0 & c & b & 0 \cr
0 & 0 & 0 & a  \end{matrix} \right)
\>
where $a(\lambda) = \sinh{(\lambda + \gamma)}$, $b(\lambda) = \sinh{(\lambda)}$ and $c(\lambda) = \sinh{(\gamma)}$. 
The $\mathcal{R}$-matrix (\ref{rmat}) satisfies the Yang-Baxter equation
\[
\label{yb}
\mathcal{R}_{12} (\lambda_1 - \lambda_2) \mathcal{R}_{13} (\lambda_1 - \lambda_3) \mathcal{R}_{23} (\lambda_2 - \lambda_3) = \mathcal{R}_{23} (\lambda_2 - \lambda_3) \mathcal{R}_{13} (\lambda_1 - \lambda_3) \mathcal{R}_{12} (\lambda_1 - \lambda_2)
\]
in $\mbox{End}(\mathbb{V}_1 \otimes \mathbb{V}_2 \otimes \mathbb{V}_3)$ ensuring the associativity of (\ref{yb_alg}).

\medskip

\paragraph{Representations.} The representations of $A(\lambda)$, $B(\lambda)$, $C(\lambda)$ and 
$D(\lambda)$ in the tensor product space $\mathbb{V}_1 \otimes \dots \otimes \mathbb{V}_L$ can be 
built as follows. We consider an ordered product of $\mathcal{R}$-matrices in the tensor
product space $\mathbb{V}_a \otimes \mathbb{V}_1 \otimes \dots \otimes \mathbb{V}_L$, namely 
\[
\label{rep}
\widebar{\mathcal{T}}_a (\lambda) = \mathop{\overrightarrow\prod}\limits_{1 \le j \le L } \mathcal{R}_{a j} (\lambda - \mu_j) \; ,
\]
with parameters $\lambda , \mu_j \in \mathbb{C}$. Due to (\ref{yb}) one can show that $\widebar{\mathcal{T}}_a$ satisfy
(\ref{yb_alg}) and thus $\mathcal{T} = \widebar{\mathcal{T}}_a$ yields a representation of
$A(\lambda)$, $B(\lambda)$, $C(\lambda)$ and $D(\lambda)$.

\medskip

\paragraph{Highest weight vector.} The vector $\ket{0}$ in (\ref{psi}) is the $\alg{sl}_2$ highest 
weight vector and it explicitly reads
\[
\ket{0} = \bigotimes_{j=1}^{L} \left( \begin{matrix} 1 \cr 0 \end{matrix} \right)
\]
in the space $\mathbb{V}_1 \otimes \dots \otimes \mathbb{V}_L$. Taking into consideration
the representation (\ref{rep}) with $\mathcal{R}$-matrix given by (\ref{rmat}), we readily obtain the properties
\begin{align}
\label{action}
A(\lambda) \ket{0} &= \prod_{j=1}^{L} a(\lambda - \mu_j) \ket{0}&  \bra{0} A(\lambda) &= \prod_{j=1}^{L} a(\lambda - \mu_j) \bra{0} \nonumber \\
D(\lambda) \ket{0} &= \prod_{j=1}^{L} b(\lambda - \mu_j) \ket{0}&  \bra{0} D(\lambda) &= \prod_{j=1}^{L} b(\lambda - \mu_j) \bra{0} \nonumber \\
C(\lambda) \ket{0} &= 0 &  \bra{0} B(\lambda) &= 0 \; ,
\end{align}
where $\bra{0}$ corresponds to the transposition of $\ket{0}$.

\medskip

\paragraph{Twisted transfer matrix.} The Yang-Baxter algebra (\ref{yb_alg}) enables us to show
that the matrix $T(\lambda) = \phi_1 A(\lambda) + \phi_2 D(\lambda)$ forms an one-parameter
family of mutually commuting matrices, i.e. $\left[ T(\lambda_1) , T(\lambda_2) \right] = 0$. 
The matrix $T(\lambda)$ is called twisted transfer matrix and for vertex models it represents
the configuration of rows spanning a two-dimensional lattice with periodic boundary conditions.
In their turn, the parameters $\phi_1 , \phi_2 \in \mathbb{C}$ govern the deviations
from strict toroidal boundary conditions and, in this sense, they introduce the notion of
twisted boundary conditions \cite{deVega_1984}. The vector $\ket{\psi}$ as defined by (\ref{psi})
is an eigenvector of $T(\lambda)$ with eigenvalue
\[
\label{eigen}
\Lambda (\lambda) = \phi_1 \prod_{j=1}^{L} a(\lambda - \mu_j) \prod_{i=1}^{n} \frac{a(\lambda_i^B - \lambda)}{b(\lambda_i^B - \lambda)} + \phi_2 \prod_{j=1}^{L} b(\lambda - \mu_j) \prod_{i=1}^{n} \frac{a(\lambda - \lambda_i^B)}{b(\lambda - \lambda_i^B)}
\]
for particular choices of the parameters $\lambda_i^B$. 

\medskip

\paragraph{Dual Bethe vector.} The dual vector 
\[
\label{dpsi}
\bra{\psi} = \bra{0} C(\lambda_1^B) \dots C(\lambda_n^B)
\]
is an eigenvector of $T$, i.e. $\bra{\psi} T(\lambda) = \Lambda(\lambda) \bra{\psi}$
with the same eigenvalue $\Lambda$ as given by (\ref{eigen}). Nevertheless, here we shall 
consider the dual vector
\[
\label{tpsi}
\bra{\widetilde{\psi}} = \bra{0} C(\lambda_1^C) \dots C(\lambda_n^C)
\] 
in order to keep our results as general as possible.

\medskip

\paragraph{Off-shell scalar product.} The scalar product of Bethe vectors $S_n$ is then defined as
\[
\label{scp}
S_n (\lambda_1^C , \dots , \lambda_n^C | \lambda_1^B , \dots , \lambda_n^B  ) = \bra{0} \prod_{i=1}^{n} C(\lambda_i^C) \; \prod_{i=1}^{n} B(\lambda_i^B) \ket{0} \; .
\]
We shall refer to (\ref{scp}) as off-shell scalar product when the variables $\lambda_i^B$ and $\lambda_i^C$
are free to assume any value on the complex plane.

\medskip

\paragraph{On-shell scalar product.} The vector $\ket{\psi}$ defined in (\ref{psi}) is an eigenvector of
the twisted transfer matrix $T(\lambda)$ only when the variables $\lambda_i^B$ satisfy certain constraints.
Thus the parameters $\lambda_i^B$ and $\lambda_i^C$ need to be fine tuned in order to having
(\ref{scp}) describing the norm of a Bethe vector. The constraints on the variables $\lambda_i^B$ ensuring that
(\ref{psi}) is a transfer matrix eigenvector are the so called Bethe ansatz equations. 
Similar constraints would be required for the dual eigenvector (\ref{tpsi}) but here we shall consider
a less restrictive condition by keeping the variables $\lambda_i^C$ arbitrary. In this way the scalar
product (\ref{scp}) under the condition
\[
\label{BA}
\prod_{j=1}^{L} \frac{a(\lambda_i^B - \mu_j)}{b(\lambda_i^B - \mu_j)} = (-1)^{n-1} \frac{\phi_2}{\phi_1} \prod_{\stackrel{k=1}{k \neq i}}^{n}  \frac{a(\lambda_i^B - \lambda_k^B)}{a(\lambda_k^B - \lambda_i^B)}
\]
will be referred to as on-shell scalar product.

\section{Functional equations}
\label{sec:FN}

This section is devoted to the derivation of functional equations describing the 
scalar product $S_n$ defined in (\ref{scp}). The main ingredient of our derivation
is the commutation rules encoded in the relation (\ref{yb_alg}) commonly referred to 
as Yang-Baxter algebra. We shall obtain two different functional equations and the
determination of $S_n$ will rely on the resolution of this system of equations. 
This method consists of an extension of the one originally proposed in 
\cite{Galleas10} for the partition function of the six-vertex model with domain
wall boundaries and subsequently generalised in \cite{Galleas_2011, Galleas_2012} 
for SOS models. Similarly to those cases, the derivation of functional equations
for scalar products will explore a consistency relation between the $\alg{sl}_2$ algebra
highest weight representation theory and the Yang-Baxter algebra.

\subsection{Equation type A}
\label{sec:eqtA}

We consider the quantity 
\[
\label{baitA}
\bra{0} \prod_{i=1}^{n} C(\lambda_j^{C}) A(\lambda_0) \prod_{i=1}^{n} B(\lambda_j^{B}) \ket{0}
\]
computed in two different ways. In the first way we consider the commutation rules 
\<
\label{commutAB}
B(\lambda_1) B(\lambda_2) &=& B(\lambda_2) B(\lambda_1) \nonumber \\
A(\lambda_1) B(\lambda_2) &=& \frac{a(\lambda_2 - \lambda_1)}{b(\lambda_2 - \lambda_1)} B(\lambda_2) A(\lambda_1) - \frac{c(\lambda_2 - \lambda_1)}{b(\lambda_2 - \lambda_1)} B(\lambda_1) A(\lambda_2)
\>
contained in the relation (\ref{yb_alg}) to move the operator $A(\lambda_0)$ in
(\ref{baitA}) to the right through the string of operators $B(\lambda_j^B)$. Due to (\ref{commutAB}),
at the last step we shall need to compute the action of a given operator $A(\lambda)$
on the state $\ket{0}$. For that we use the relation (\ref{action}) arising from the
$\alg{sl}_2$ highest weight representation theory. By doing so we find that the quantity 
(\ref{baitA}) consists of a linear combination of terms $S_n$ with $B$-type arguments in
the set $\{ \lambda_0 , \lambda_1^{B} , \dots , \lambda_n^{B} \}$ with only $n$ elements being
taken at a time.

The second way of evaluating (\ref{baitA}) is by moving the operator $A(\lambda_0)$ to the left
through all the operators $C(\lambda_j^{C})$. This can be implemented with the help of the relations
\<
\label{commutAC}
C(\lambda_1) C(\lambda_2) &=& C(\lambda_2) C(\lambda_1) \nonumber \\
C(\lambda_1) A(\lambda_2) &=& \frac{a(\lambda_1 - \lambda_2)}{b(\lambda_1 - \lambda_2)} A(\lambda_2) C(\lambda_1) - \frac{c(\lambda_1 - \lambda_2)}{b(\lambda_1 - \lambda_2)} A(\lambda_1) C(\lambda_2) \; ,
\>
which are also among the ones encoded in (\ref{yb_alg}). Then at the last stage we will need the
quantity $\bra{0} A(\lambda)$ which is given in (\ref{action}). Thus this alternative route
of computing (\ref{baitA}) also yields a linear combination of terms $S_n$ but this time with
$C$-type arguments in the set $\{ \lambda_0 , \lambda_1^{C} , \dots , \lambda_n^{C} \}$ 
where only $n$ variables are taken at a time.

In this way the consistency between these two routes of computing (\ref{baitA})
implies the functional equation
\<
\label{typeA}
&& M_0 \; S_n (\lambda_1^{C} , \dots , \lambda_n^{C} | \lambda_1^{B} , \dots , \lambda_n^{B} ) + \sum_{i=1}^{n} N_i^{(B)} S_n (\lambda_1^{C} , \dots , \lambda_n^{C} | \lambda_0 , \lambda_{1}^{B} ,  \dots , \lambda_{i-1}^{B}, \lambda_{i+1}^{B}, \dots , \lambda_n^{B} ) \nonumber \\
&& + \sum_{i=1}^{n} N_i^{(C)} S_n (\lambda_0 , \lambda_{1}^{C} , \dots , \lambda_{i-1}^{C}, \lambda_{i+1}^{C}, \dots , \lambda_n^{C} | \lambda_1^{B} , \dots , \lambda_n^{B} ) = 0 \; , \nonumber \\
\>
with coefficients
\<
\label{coeffA}
M_0 &=& \prod_{j=1}^{L} a(\lambda_0 - \mu_j) \left[  \prod_{i=1}^{n} \frac{a(\lambda_i^{C} - \lambda_0)}{b(\lambda_i^{C} - \lambda_0)} - \prod_{i=1}^{n} \frac{a(\lambda_i^{B} - \lambda_0)}{b(\lambda_i^{B} - \lambda_0)} \right] \nonumber \\
N_i^{(B, C)} &=& \alpha_{B, C} \frac{c(\lambda_i^{B,C} - \lambda_0)}{b(\lambda_i^{B,C} - \lambda_0)} \prod_{j=1}^{L} a(\lambda_i^{B, C} - \mu_j) \prod_{j \neq i}^{n} \frac{a(\lambda_j^{B, C} - \lambda_i^{B, C})}{b(\lambda_j^{B, C} - \lambda_i^{B, C})} \; ,
\>
where $\alpha_B =1$ and $\alpha_C = - 1$.

\medskip

\subsection{Equation type D}
\label{sec:eqtD}

The same mechanism employed in \Secref{sec:eqtA} can also be considered
starting with the quantity
\[
\label{baitD}
\bra{0} \prod_{i=1}^{n} C(\lambda_j^{C}) D(\lambda_0) \prod_{i=1}^{n} B(\lambda_j^{B}) \ket{0}
\]
instead of (\ref{baitA}). In that case, moving the operator $D(\lambda_0)$ to the right through the
string of operators $B(\lambda_j^{B})$ will require the use of the Yang-Baxter algebra relation
\<
D(\lambda_1) B(\lambda_2) &=& \frac{a(\lambda_1 - \lambda_2)}{b(\lambda_1 - \lambda_2)} B(\lambda_2) D(\lambda_1) - \frac{c(\lambda_1 - \lambda_2)}{b(\lambda_1 - \lambda_2)} B(\lambda_1) D(\lambda_2) \; .
\>
On the other hand, the commutation rule
\<
C(\lambda_1) D(\lambda_2) &=& \frac{a(\lambda_2 - \lambda_1)}{b(\lambda_2 - \lambda_1)} D(\lambda_2) C(\lambda_1) - \frac{c(\lambda_2 - \lambda_1)}{b(\lambda_2 - \lambda_1)} D(\lambda_1) C(\lambda_2)
\>
will be required in order to move the operator $D(\lambda_0)$ to the left through all the operators $C(\lambda_j^{C})$.
Besides that we shall also consider the commutation rules $\left[ B(\lambda_1) , B(\lambda_2) \right] = \left[ C(\lambda_1) , C(\lambda_2) \right] = 0$
and the properties (\ref{action}). 

Thus the consistency condition between these two different ways of computing
(\ref{baitD}) leave us with the following functional equation,
\<
\label{typeD}
&& \widetilde{M}_0 \; S_n (\lambda_1^{C} , \dots , \lambda_n^{C} | \lambda_1^{B} , \dots , \lambda_n^{B} ) + \sum_{i=1}^{n} \widetilde{N}_i^{(B)} S_n (\lambda_1^{C} , \dots , \lambda_n^{C} | \lambda_0 , \lambda_{1}^{B} ,  \dots , \lambda_{i-1}^{B}, \lambda_{i+1}^{B}, \dots , \lambda_n^{B} ) \nonumber \\
&& + \sum_{i=1}^{n} \widetilde{N}_i^{(C)} S_n (\lambda_0 , \lambda_{1}^{C} , \dots , \lambda_{i-1}^{C}, \lambda_{i+1}^{C}, \dots , \lambda_n^{C} | \lambda_1^{B} , \dots , \lambda_n^{B} ) = 0 \; , \nonumber \\
\>
whose coefficients are explicitly given by
\<
\label{coeffD}
\widetilde{M}_0 &=& \prod_{j=1}^{L} b(\lambda_0 - \mu_j) \left[  \prod_{i=1}^{n} \frac{a(\lambda_0 - \lambda_i^{C})}{b(\lambda_0 - \lambda_i^{C})} - \prod_{i=1}^{n} \frac{a(\lambda_0 - \lambda_i^{B})}{b(\lambda_0 - \lambda_i^{B})} \right] \nonumber \\
\widetilde{N}_i^{(B, C)} &=& \alpha_{B, C} \frac{c(\lambda_0 - \lambda_i^{B,C})}{b(\lambda_0 - \lambda_i^{B,C})} \prod_{j=1}^{L} b(\lambda_i^{B, C} - \mu_j) \prod_{j \neq i}^{n} \frac{a(\lambda_i^{B, C} - \lambda_j^{B, C})}{b(\lambda_i^{B, C} - \lambda_j^{B, C})} \; .
\>

\medskip

In summary, we have demonstrated in this section how the Yang-Baxter algebra can be explored in order to
derive functional equations for the scalar product of Bethe vectors. We have obtained two distinct
equations which we shall refer to as equation of type A (\ref{typeA}) and equation of type D (\ref{typeD}).
The solution of these equations will be discussed in the next section.

\section{Scalar product}
\label{sec:scalar}

Solving the system of functional equations formed by (\ref{typeA}) and (\ref{typeD})
is the main goal of this section and some remarks are required in order to proceed.
For instance, the method employed here for the derivation of (\ref{typeA}) and (\ref{typeD})
can be seen as an extension of the method considered in \cite{Galleas_2012}, and
the resulting functional relations indeed share some similarities with the one obtained
for the partition function of the elliptic SOS model with domain wall boundaries.
However, there are still some important structural differences that introduce some extra steps
in solving (\ref{typeA}) and (\ref{typeD}). 

Firstly, let us consider the similarities. The functional equations obtained here and the one
of \cite{Galleas_2012} are all relations for a multivariate function, i.e. $F(z_1 , \dots , z_n)$,
composed of a linear combination of terms involving the function $F$ with a given variable
$z_i$ in its argument being replaced by a variable $z_0$.
Moreover, the Eqs. (\ref{typeA}) and (\ref{typeD}) are also homogeneous in the sense that
$\alpha S _n$ solves our system of equations if $S_n$ is a solution and $\alpha$ is a constant. 
In fact, $\alpha$ only needs to be independent of the variables $\lambda_i^B$ and 
$\lambda_i^C$. This property tells us in advance that Eqs. (\ref{typeA}) and (\ref{typeD}) will
be able to determine $S_n$ only up to an overall multiplicative factor independent of
$\lambda_i^{B}$ and $\lambda_i^{C}$ at most. 
Thus it will be necessary to evaluate the function $S_n$ for a particular value of its variables
in order to have our scalar product completely fixed. 
Still considering the similarities, both Eqs. (\ref{typeA}) and (\ref{typeD}) are linear which raise
the issue of uniqueness of the solution. Here we will be interested in a multivariate polynomial solution
and this property ensures uniqueness as demonstrated in \cite{Galleas_2011} under very general conditions.

Now let us consider the differences between (\ref{typeA}, \ref{typeD}) and the functional equation
obtained in \cite{Galleas_2012}. The most obvious difference is that here we have obtained two equations which
might suggest that one of them is redundant. However, the direct inspection of our equations
for small values of $n$ and $L$ reveals that the polynomial solution which we shall be interested
can not be completely fixed by only one of the equations. The situation is different when we consider
both equations simultaneously, and their direct inspection shows that the system (\ref{typeA}, \ref{typeD})
indeed determines the polynomial solution up to an overall multiplicative factor.    
Furthermore, in the case considered in \cite{Galleas_2012} we have an equation running only over one
set of variables, i.e. $\{ \lambda_0 , \lambda_1 , \dots , \lambda_n \}$.
Here both of our equations run over the two sets of variables $\{ \lambda_0 , \lambda_1^{B} , \dots , \lambda_n^{B} \}$ 
and $\{ \lambda_0 , \lambda_1^{C} , \dots , \lambda_n^{C} \}$. Taking into account the above discussion, the following Lemmas 
will pave the way for solving (\ref{typeA}, \ref{typeD}).

\begin{lemma}[Polynomial structure] \label{polst}
In terms of variables $x_i^{B,C} = e^{2 \lambda_i^{B,C}}$, the function $S_n$ is of the form
$S_n  = \prod_{i=1}^{n} (x_i^B x_i^C )^{-\frac{L-1}{2}} \bar{S}_n (x_1^C , \dots , x_n^C | x_1^B , \dots , x_n^B )$
where $\bar{S}_n$ is a polynomial of order $L-1$ in each one of its variables separately. 
\end{lemma}
\begin{proof}
See \Appref{sec:POL}. 
\end{proof}

\begin{lemma}[Special zeroes] \label{zero}
The function $S_n (\lambda_1^{C} , \dots , \lambda_n^{C} | \lambda_1^{B} , \dots , \lambda_n^{B})$
vanishes for the specialisation of variables $\lambda_1^{B} = \mu_1$ and $\lambda_2^{B} = \mu_1 - \gamma$.
The same property also holds for the specialisation $\lambda_1^{C} = \mu_1$ and $\lambda_2^{C} = \mu_1 - \gamma$.
\end{lemma}
\begin{proof}
See \Appref{sec:ZERO}. 
\end{proof}

\begin{lemma}[Doubly symmetric function] \label{symm}
The scalar product $S_n ( \lambda_1^{C} , \dots , \lambda_n^{C} | \lambda_1^{B} , \dots , \lambda_n^{B})$
is a symmetric function in each one of the set of variables $\{ \lambda_i^{B}  \}$ and $\{ \lambda_i^{C} \}$
independently. More precisely, 
\<
&& S_n ( \lambda_1^{C} , \dots , \lambda_i^{C} , \dots , \lambda_j^{C} , \dots , \lambda_n^{C} | \lambda_1^{B} , \dots , \lambda_n^{B}) \nonumber \\
&&  \qquad \qquad \qquad \qquad \qquad \qquad \qquad \qquad = S_n ( \lambda_1^{C} , \dots , \lambda_j^{C} , \dots , \lambda_i^{C} , \dots , \lambda_n^{C} | \lambda_1^{B} , \dots , \lambda_n^{B}) \nonumber
\>
and 
\<
&& S_n ( \lambda_1^{C} , \dots , \lambda_n^{C} | \lambda_1^{B} , \dots , \lambda_i^{B} , \dots , \lambda_j^{B} , \dots , \lambda_n^{B}) \nonumber \\
&&  \qquad \qquad \qquad \qquad \qquad \qquad \qquad \qquad = S_n ( \lambda_1^{C} , \dots , \lambda_n^{C} | \lambda_1^{B} , \dots , \lambda_j^{B} , \dots , \lambda_i^{B} , \dots , \lambda_n^{B} ) \; . \nonumber
\>
\end{lemma}
\begin{proof}
See \Appref{sec:sym}.
\end{proof}

\begin{lemma}[Asymptotic behaviour] \label{asymp}
In the limit $x_i^{B,C} \rightarrow \infty$, the function $S_n$ behaves as
\begin{equation}
\label{asy}
S_n  \sim  \frac{(q - q^{-1})^{2n}}{2^{2 n L}} q^{n(L-n)} [ n! ]_{q^2}^2 e^{-2 n \sum_{j=1}^{L} \mu_j}  \sum_{1 \leq a_1 < \dots < a_n \leq L} e^{2 \sum_{j=1}^{n} \mu_{a_j}} \prod_{i=1}^{n} (x_i^B x_i^C)^{\frac{L-1}{2}}  \; , \nonumber \\
\end{equation}
where $[ n! ]_{q^2}$ denotes the $q$-factorial function defined as
\<
[ n! ]_{q^2} = 1 (1 + q^2)(1+ q^2 + q^4) \dots (1+ q^2 + \dots + q^{2(n-1)}) \; .
\>
\end{lemma}
\begin{proof}
See \Appref{sec:ASYMP}. 
\end{proof}

\begin{remark}
Due to the Lemma \ref{symm} we can safely employ the notation 
\[
S_n (\lambda_1^C , \dots , \lambda_n^C | \lambda_1^B , \dots , \lambda_n^B) = S_n (X^{1,n} | Y^{1,n}) \nonumber 
\]
where $X^{i,j} = \{ \lambda_k^C \; | \; i \leq k \leq j \}$ and $Y^{i,j} = \{ \lambda_k^B \; | \; i \leq k \leq j \}$. 
\end{remark}

Except for the Lemmas \ref{polst} and \ref{asymp}, the remaining ones are a direct consequence of the
functional relations (\ref{typeA}) and (\ref{typeD}).

\subsection{Off-shell formula}
\label{sec:offshell}

Here we shall consider the resolution of the system of equations (\ref{typeA}, \ref{typeD})
for general values of variables $\lambda_i^B$ and $\lambda_i^C$. This case is refereed to as
off-shell scalar product as remarked in \Secref{sec:bethe}. The methodology we shall employ for 
solving (\ref{typeA}, \ref{typeD}) is similar to the one developed in \cite{Galleas_2012}
and in what follows we describe a sequence of steps leading to the solution.

\paragraph{Step 1.} We firstly consider Eq. (\ref{typeA}) under the specialisation of variables
$\lambda_0 = \mu_1 - \gamma$ and $\lambda_n^B = \mu_1$ such that the coefficient $M_0$ vanishes.
We shall also consider the property $S_n ( X^{1,n} | \mu_1 - \gamma , \dots , \mu_1 ) = 0$
obtained from Lemmas \ref{zero} and \ref{symm} implying the relation
\<
\label{sv}
S_n (\bar{X}^{2,n} | \check{Y}^{2,n} ) &=& \prod_{j=2}^{n} b(\lambda_j^C - \mu_1)  a(\lambda_j^B - \mu_1) \; V( X^{2,n} | Y^{2,n} ) \; ,
\>
due to Lemmas  \ref{polst} and \ref{symm}. In (\ref{sv}) we have also introduced the notation
$\bar{Z}^{i,j} = Z^{i,j} \cup \{ \mu_1 - \gamma  \}$ and $\check{Z}^{i,j} = Z^{i,j} \cup \{ \mu_1 \}$
for $Z^{i,j} \in \{ X^{i,j} , Y^{i,j} \}$. In its turn the function $V$ appearing in (\ref{sv}) is also of the 
form described in Lemma \ref{polst} under the identifications $S_n \mapsto V$, $ n \mapsto n-1$ and $L \mapsto L-1$. 
Thus under this specialisation of variables, Eq. (\ref{typeA}) yields the relation 
\<
\label{sv1}
S_n (X^{1,n}  | \bar{Y}^{1,n-1}  ) = \sum_{i=1}^{n}  m_i^{(A)} \; V (X_i^{1,n} | Y^{1,n-1}) \; , 
\>
where $X_k^{i,j} = X^{i,j} \backslash \{ \lambda_k^C \}$ and
\<
\label{mia}
m_i^{(A)} &=& - \prod_{j=1}^{n-1} a(\lambda_j^{B} - \mu_1) \prod_{\stackrel{j=1}{j \neq i}}^{n} b(\lambda_j^{C} - \mu_1) \left. \frac{N_i^{(C)}}{N_n^{(B)}} \right|_{\stackrel{\lambda_0 = \mu_1 - \gamma}{\lambda_n^{B} = \mu_1}}  \nonumber \\
&=& \frac{c}{a(\lambda_i^C - \mu_1)} \prod_{j=1}^{L} \frac{a(\lambda_i^C - \mu_j)}{a(\mu_1 - \mu_j)} \prod_{j=1}^{n-1} b(\lambda_j^B - \mu_1) \prod_{\stackrel{j=1}{j \neq i}}^{n} b(\lambda_j^C - \mu_1) \frac{a(\lambda_j^C - \lambda_i^C)}{b(\lambda_j^C - \lambda_i^C)} \; . \nonumber \\
\>

\paragraph{Step 2.} Analogously to (\ref{sv}), due to Lemmas \ref{polst}, \ref{zero} and \ref{symm} we can write
\<
\label{rv}
S_n (\check{X}^{2,n} | \bar{Y}^{2,n} ) &=& \prod_{j=2}^{n} a(\lambda_j^C - \mu_1)  b(\lambda_j^B - \mu_1) \; W( X^{2,n} | Y^{2,n} ) \; ,
\>
where the function $W$ is also of the form described in Lemma \ref{polst} under the mappings $n \mapsto n-1$ and $L \mapsto L-1$.
Then by setting $\lambda_0 = \mu_1 - \gamma$ and $\lambda_n^C = \mu_1$ in Eq. (\ref{typeA}), and considering 
the relation (\ref{rv}) in addition to the property $S_n (\mu_1 - \gamma , \dots , \mu_1 | Y^{1,n}) = 0$, we obtain the formula
\<
\label{rv1}
S_n (\bar{X}^{1,n-1} | Y^{1,n}  ) = \sum_{i=1}^{n}  \bar{m}_i^{(A)} \; W (X^{1,n-1} | Y_i^{1,n}) \; , 
\>
where $Y_k^{i,j} = Y^{i,j} \backslash \{ \lambda_k^B \}$ and
\<
\label{bmia}
\bar{m}_i^{(A)} &=& - \prod_{j=1}^{n-1} a(\lambda_j^{C} - \mu_1) \prod_{\stackrel{j=1}{j \neq i}}^{n} b(\lambda_j^{B} - \mu_1) \left. \frac{N_i^{(B)}}{N_n^{(C)}} \right|_{\stackrel{\lambda_0 = \mu_1 - \gamma}{\lambda_n^{C} = \mu_1}}  \nonumber \\
&=& \frac{c}{a(\lambda_i^B - \mu_1)} \prod_{j=1}^{L} \frac{a(\lambda_i^B - \mu_j)}{a(\mu_1 - \mu_j)} \prod_{j=1}^{n-1} b(\lambda_j^C - \mu_1) \prod_{\stackrel{j=1}{j \neq i}}^{n} b(\lambda_j^B - \mu_1) \frac{a(\lambda_j^B - \lambda_i^B)}{b(\lambda_j^B - \lambda_i^B)} \; . \nonumber \\
\>

\paragraph{Step 3.} The coefficient $N_n^{(B)}$ vanishes under the specialisation $\lambda_n^{B} = \mu_1 - \gamma$.
Thus, under this particular specialisation, the Eq. (\ref{typeA}) will contain only terms of the form $S_n ( X^{1,n} | \bar{Y}^{1,n-1})$
allowing us to use the relation (\ref{sv1}) to obtain an equation involving solely the function $V$.
The resulting equation can be further simplified by setting $\lambda_n^{C} = \mu_1$, and in this way we are left with
the relation
\<
\label{typeAV}
J_0 \; V ( X^{1,n-1} | Y^{1,n-1} ) + \sum_{i=1}^{n-1} K_i^{(B)} V ( X^{1,n-1} | Y_i^{0,n-1}) + \sum_{i=1}^{n-1} K_i^{(C)} V ( X_i^{0,n-1} | Y^{1,n-1}) = 0 \; .  \nonumber \\
\>
In their turn the coefficients $J_0$ and $K_i^{(B,C)}$ appearing in (\ref{typeAV}) correspond respectively to the coefficients $M_0$
and $N_i^{(B,C)}$ given in (\ref{coeffA}) under the mappings $L \mapsto L-1$, $n \mapsto n-1$ and 
$\mu_i \mapsto \mu_{i+1}$. Thus the function $V$ obeys essentially the same equation as $S_{n-1}$ for 
Bethe vectors living in $\mathbb{V}_1 \otimes \dots \otimes \mathbb{V}_{L-1}$.

\paragraph{Step 4.} Next we consider Eq. (\ref{typeA}) with $\lambda_n^C = \mu_1 - \gamma$. In that case the 
coefficient $N_n^{(C)}$ vanishes and we are left only with terms of the form $S_n (\bar{X}^{1,n-1} | Y^{1,n} )$. 
Then we use the relation (\ref{rv1}) to obtain an equation only in terms of the function $W$. After setting 
$\lambda_n^B = \mu_1$ and eliminating an overall factor, the equation obtained in this way reads
\<
\label{typeAW}
J_0 \; W ( X^{1,n-1} | Y^{1,n-1} ) + \sum_{i=1}^{n-1} K_i^{(B)} W ( X^{1,n-1} | Y_i^{0,n-1}) + \sum_{i=1}^{n-1} K_i^{(C)} W ( X_i^{0,n-1} | Y^{1,n-1}) = 0 \; . \nonumber \\
\>  
Thus the functions $V$ and $W$ obey the same equation and the uniqueness of the solution of (\ref{typeA}) would
imply that $V$ and $W$ can differ only by an overall constant factor.

\paragraph{Step 5.} We consider the Eq. (\ref{typeD}) with  $\lambda_0 = \mu_1$ and $\lambda_n^B = \mu_1 - \gamma$.
In that case the coefficient $\widetilde{M}_0 = 0$, and we can use formula (\ref{rv}) in addition to the 
Lemma \ref{zero} to obtain the relation
\<
\label{tv1}
S_n ( X^{1,n} | \check{Y}^{1,n-1} ) = \sum_{i=1}^{n}  \bar{m}_i^{(D)} \; W (X_i^{1,n} | Y^{1,n-1}) \; , 
\>
where
\<
\label{bmid}
\bar{m}_i^{(D)} &=& - \prod_{j=1}^{n-1} b(\lambda_j^B - \mu_1) \prod_{\stackrel{j=1}{j \neq i}}^{n} a(\lambda_j^C - \mu_1) \left. \frac{\widetilde{N}_i^{(C)}}{\widetilde{N}_n^{(B)}} \right|_{\stackrel{\lambda_0 = \mu_1}{\lambda_n^B = \mu_1 - \gamma}} \nonumber \\
&=& \frac{c}{b(\mu_1 - \lambda_i^C)} \prod_{j=1}^{L} \frac{b(\mu_j - \lambda_i^C)}{a(\mu_j - \mu_1)} 
\prod_{j=1}^{n-1} a(\lambda_j^B - \mu_1) \prod_{\stackrel{j=1}{j \neq i}}^{n} a(\lambda_j^C - \mu_1) \frac{a(\lambda_i^C - \lambda_j^C)}{b(\lambda_i^C - \lambda_j^C)} \; . \nonumber \\
\>

\paragraph{Step 6.} Next we set $\lambda_0 = \mu_1$ and $\lambda_n^C = \mu_1 - \gamma$ in the Eq. (\ref{typeD}).
This allows us to employ the relation (\ref{sv}) and the property $S_n (\mu_1 , \dots , \mu_1 - \gamma | Y^{1,n}) = 0$
described in Lemma \ref{zero}. In this way we obtain the expression 
\<
\label{uv}
S_n (\check{X}^{1,n-1} | Y^{1,n}) = \sum_{i=1}^n m_i^{(D)} \; V (X^{1,n-1} | Y_i^{1,n}) \; , 
\>
with
\<
\label{mid}
m_i^{(D)} &=& - \prod_{j=1}^{n-1} b(\lambda_j^C - \mu_1) \prod_{\stackrel{j=1}{j \neq i}}^{n} a(\lambda_j^B - \mu_1) \left. \frac{\widetilde{N}_i^{(B)}}{\widetilde{N}_n^{(C)}} \right|_{\stackrel{\lambda_0 = \mu_1}{\lambda_n^C = \mu_1 - \gamma}} \nonumber \\
&=& \frac{c}{b(\mu_1 - \lambda_i^B)} \prod_{j=1}^{L} \frac{b(\mu_j - \lambda_i^B)}{a(\mu_j - \mu_1)} 
\prod_{j=1}^{n-1} a(\lambda_j^C - \mu_1) \prod_{\stackrel{j=1}{j \neq i}}^{n} a(\lambda_j^B - \mu_1) \frac{a(\lambda_i^B - \lambda_j^B)}{b(\lambda_i^B - \lambda_j^B)} \; . \nonumber \\
\>

\paragraph{Step 7.} The Eq. (\ref{typeD}) will contain only terms of the form $S_n ( X^{1,n} | \check{Y}^{1,n-1} )$
for the specialisation $\lambda_n^{B} = \mu_1$ since the coefficient $\widetilde{N}_n^{(B)}$ vanishes.
In that case we can consider the formula (\ref{tv1}) and also set $\lambda_n^{C} = \mu_1 - \gamma$. By doing so
we are left with the relation,
\<
\label{typeDW}
\widetilde{J}_0 \; W ( X^{1,n-1} | Y^{1,n-1} ) + \sum_{i=1}^{n-1} \widetilde{K}_i^{(B)} W ( X^{1,n-1} | Y_i^{0,n-1}) + \sum_{i=1}^{n-1} \widetilde{K}_i^{(C)} W ( X_i^{0,n-1} | Y^{1,n-1}) = 0 \; , \nonumber \\
\>  
whose coefficients $\widetilde{J}_0$ and $\widetilde{K}_i^{(B,C)}$ correspond respectively to the coefficients
$\widetilde{M}_0$ and $\widetilde{N}_i^{(B,C)}$ given in (\ref{coeffD}) under the mappings 
$L \mapsto L-1$, $n \mapsto n-1$ and $\mu_i \mapsto \mu_{i+1}$.

\paragraph{Step 8.} Set $\lambda_n^C = \mu_1$ in Eq. (\ref{typeD}) taking into account that the coefficient
$\widetilde{N}_n^{(C)}$ vanishes. For this particular specialisation Eq. (\ref{typeD}) contains only
terms of the form $S_n (\check{X}^{1,n-1} | Y^{1,n})$. In this way we use formula (\ref{uv}) and set
$\lambda_n^B = \mu_1 - \gamma$ to obtain the relation
\<
\label{typeDV}
\widetilde{J}_0 \; V ( X^{1,n-1} | Y^{1,n-1} ) + \sum_{i=1}^{n-1} \widetilde{K}_i^{(B)} V ( X^{1,n-1} | Y_i^{0,n-1}) + \sum_{i=1}^{n-1} \widetilde{K}_i^{(C)} V ( X_i^{0,n-1} | Y^{1,n-1}) = 0 \; . \nonumber \\
\>  
Thus, considering Eq. (\ref{typeDW}) in addition to (\ref{typeDV}), we can see that both functions $V$ and $W$ also obeys
the Eq. (\ref{typeD}) under the mappings $L \mapsto L-1$, $n \mapsto n-1$ and $\mu_i \mapsto \mu_{i+1}$.

\paragraph{Step 9.} Next we consider Eq. (\ref{typeD}) with $\lambda_0 = \mu_1-\gamma$ and make use of the relations
(\ref{rv1}) and (\ref{sv1}). By doing so we obtain the expression
\<
\label{reA}
S_n ( X^{1,n} | Y^{1,n} ) =  \sum_{i,j = 1}^n \frac{\Theta_{i,j}}{\mathcal{F}} \; V( X_i^{1,n} | Y_j^{1,n} ) +  \sum_{i,j = 1}^n \frac{\Phi_{i,j}}{\mathcal{F}} \; W( X_j^{1,n} | Y_i^{1,n} ) \; , 
\> 
where
\<
\label{thetaphiA}
\Theta_{i,j} &=&   \frac{c^2}{b(\lambda_i^C - \mu_1) b(\mu_1 -  \lambda_j^B)} \prod_{k=1}^{L} \frac{a(\lambda_i^C - \mu_k) b(\mu_k - \lambda_j^B)}{a(\mu_1 - \mu_k) a(\mu_k - \mu_1)} \nonumber \\
&& \times \prod_{\stackrel{k=1}{k \neq i}}^{n} a(\lambda_k^C - \mu_1) \frac{a(\lambda_k^C - \lambda_i^C)}{b(\lambda_k^C - \lambda_i^C)} \prod_{\stackrel{k=1}{k \neq j}}^{n} a(\lambda_k^B - \mu_1) \frac{a(\lambda_j^B - \lambda_k^B)}{b(\lambda_j^B - \lambda_k^B)} \nonumber \\
\Phi_{i,j} &=&  \frac{c^2}{b(\lambda_i^B - \mu_1) b(\lambda_j^C - \mu_1)} \prod_{k=1}^{L} \frac{a(\lambda_i^B - \mu_k) b(\mu_k - \lambda_j^C)}{a(\mu_1 - \mu_k) a(\mu_k - \mu_1)} \nonumber \\
&& \prod_{\stackrel{k=1}{k \neq i}}^{n} a(\lambda_k^B - \mu_1) \frac{a(\lambda_k^B - \lambda_i^B)}{b(\lambda_k^B - \lambda_i^B)} \prod_{\stackrel{k=1}{k \neq j}}^{n} a(\lambda_k^C - \mu_1) \frac{a(\lambda_j^C - \lambda_k^C)}{b(\lambda_j^C - \lambda_k^C)}  \nonumber \\
\mathcal{F} &=& \prod_{k=1}^{n} \frac{a(\lambda_k^C - \mu_1)}{b(\lambda_k^C - \mu_1)} - \prod_{k=1}^{n} \frac{a(\lambda_k^B - \mu_1)}{b(\lambda_k^B - \mu_1)} \;\; . \nonumber \\
\>

\paragraph{Step 10.} For completeness we also set $\lambda_0 = \mu_1$ in Eq. (\ref{typeA}) and consider the expressions
(\ref{tv1}) and (\ref{uv}). This procedure yields the same formula (\ref{reA}).
\bigskip
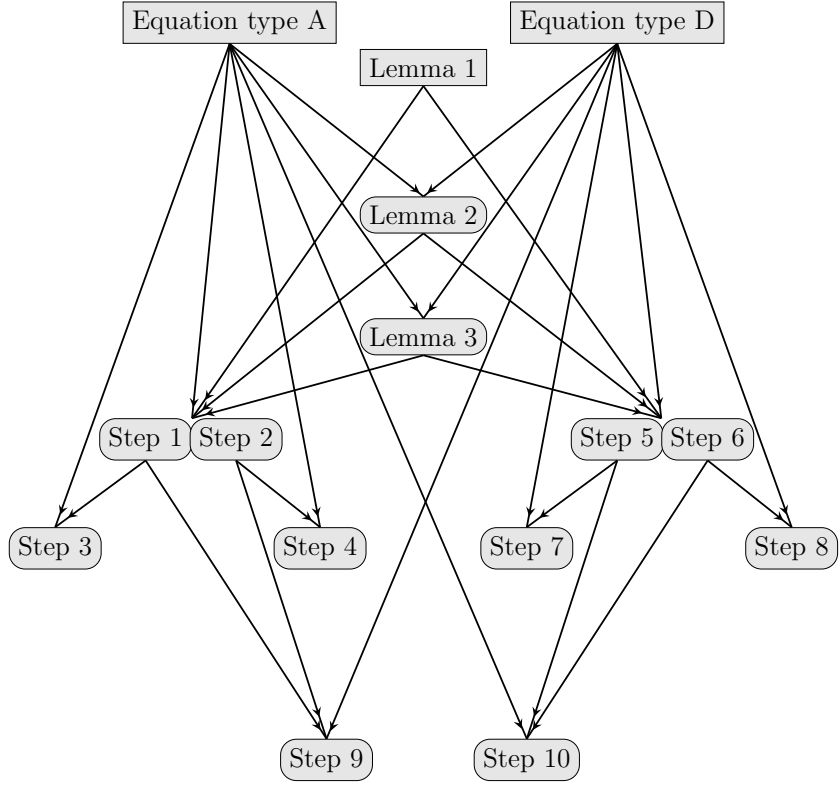
\begin{figure} \centering
\scalebox{0.85}{
\begin{tikzpicture}[>=stealth]
\path (0,0) node[rectangle,fill=gray!20!white,draw] (p1) {Equation type A}
      (6,0) node[rectangle,fill=gray!20!white,draw] (p2) {Equation type D};
\begin{scope}[yshift=-4.3cm]
\path (3,-0.6) node[rectangle,rounded corners=7pt,fill=gray!20!white,draw] (p3) {Lemma \ref{symm}}
      (3,1.3) node[rectangle,rounded corners=7pt,fill=gray!20!white,draw] (p4) {Lemma \ref{zero}}
      (3,3.6) node[rectangle,fill=gray!20!white,draw] (p5) {Lemma \ref{polst}};
\end{scope}
\begin{scope}[yshift=-6.5cm]
\path (-1.3,0) node[rectangle,rounded corners=7pt,fill=gray!20!white,draw] (p6) {Step $1$}
      (0.1,0) node[rectangle,rounded corners=7pt,fill=gray!20!white,draw] (p7) {Step $2$};
\end{scope}
\begin{scope}[yshift=-6.5cm, xshift=7.3cm]
\path (-1.3,0) node[rectangle,rounded corners=7pt,fill=gray!20!white,draw] (p8) {Step $5$}
      (0.1,0) node[rectangle,rounded corners=7pt,fill=gray!20!white,draw] (p9) {Step $6$};
\end{scope}
\begin{scope}[yshift=-8.2cm]
\path (-2.7,0) node[rectangle,rounded corners=7pt,fill=gray!20!white,draw] (p10) {Step $3$}
      (1.4,0) node[rectangle,rounded corners=7pt,fill=gray!20!white,draw] (p11) {Step $4$};
\end{scope}
\begin{scope}[yshift=-8.2cm, xshift=7.3cm]
\path (-2.7,0) node[rectangle,rounded corners=7pt,fill=gray!20!white,draw] (p12) {Step $7$}
      (1.4,0) node[rectangle,rounded corners=7pt,fill=gray!20!white,draw] (p13) {Step $8$};
\end{scope}
\begin{scope}[yshift=-11.5cm]
\path (1.5,0) node[rectangle,rounded corners=7pt,fill=gray!20!white,draw] (p14) {Step $9$};
\end{scope}
\begin{scope}[yshift=-11.5cm, xshift=7.3cm]
\path (-2.7,0) node[rectangle,rounded corners=7pt,fill=gray!20!white,draw] (p15) {Step $10$};
\end{scope}
\draw [postaction=decorate,decoration={markings, mark=at position 5.1cm with {\arrow[black]{stealth}}}, thick]  (p1.south) -- (p3.north);
\draw [postaction=decorate,decoration={markings, mark=at position 5.1cm with {\arrow[black]{stealth}}}, thick]  (p2.south) -- (p3.north);
\draw [postaction=decorate,decoration={markings, mark=at position 3.75cm with {\arrow[black]{stealth}}}, thick]  (p1.south) -- (p4.north);
\draw [postaction=decorate,decoration={markings, mark=at position 3.75cm with {\arrow[black]{stealth}}}, thick]  (p2.south) -- (p4.north);
\draw [postaction=decorate,decoration={markings, mark=at position 5.7cm with {\arrow[black]{stealth}}}, thick]  (p1.south) -- (p6.north east);
\draw [postaction=decorate,decoration={markings, mark=at position 3.5cm with {\arrow[black]{stealth}}}, thick]  (p3.south) -- (p6.north east);
\draw [postaction=decorate,decoration={markings, mark=at position 4.4cm with {\arrow[black]{stealth}}}, thick]  (p4.south) -- (p6.north east);
\draw [postaction=decorate,decoration={markings, mark=at position 5.95cm with {\arrow[black]{stealth}}}, thick]  (p5.south) -- (p6.north east);
\draw [postaction=decorate,decoration={markings, mark=at position 7.8cm with {\arrow[black]{stealth}}}, thick]  (p1.south) -- (p10.north);
\draw [postaction=decorate,decoration={markings, mark=at position 7.5cm with {\arrow[black]{stealth}}}, thick]  (p1.south) -- (p11.north);
\draw [postaction=decorate,decoration={markings, mark=at position 11.6cm with {\arrow[black]{stealth}}}, thick]  (p2.south) -- (p14.north);
\draw [postaction=decorate,decoration={markings, mark=at position 5.0cm with {\arrow[black]{stealth}}}, thick]  (p6.south) -- (p14.north);
\draw [postaction=decorate,decoration={markings, mark=at position 4.2cm with {\arrow[black]{stealth}}}, thick]  (p7.south) -- (p14.north);
\draw [postaction=decorate,decoration={markings, mark=at position 1.5cm with {\arrow[black]{stealth}}}, thick]  (p6.south) -- (p10.north);
\draw [postaction=decorate,decoration={markings, mark=at position 1.5cm with {\arrow[black]{stealth}}}, thick]  (p7.south) -- (p11.north);
\draw [postaction=decorate,decoration={markings, mark=at position 5.7cm with {\arrow[black]{stealth}}}, thick]  (p2.south) -- (p9.north west);
\draw [postaction=decorate,decoration={markings, mark=at position 3.5cm with {\arrow[black]{stealth}}}, thick]  (p3.south) -- (p9.north west);
\draw [postaction=decorate,decoration={markings, mark=at position 4.4cm with {\arrow[black]{stealth}}}, thick]  (p4.south) -- (p9.north west);
\draw [postaction=decorate,decoration={markings, mark=at position 5.95cm with {\arrow[black]{stealth}}}, thick]  (p5.south) -- (p9.north west);
\draw [postaction=decorate,decoration={markings, mark=at position 7.8cm with {\arrow[black]{stealth}}}, thick]  (p2.south) -- (p13.north);
\draw [postaction=decorate,decoration={markings, mark=at position 7.5cm with {\arrow[black]{stealth}}}, thick]  (p2.south) -- (p12.north);
\draw [postaction=decorate,decoration={markings, mark=at position 11.6cm with {\arrow[black]{stealth}}}, thick]  (p1.south) -- (p15.north);
\draw [postaction=decorate,decoration={markings, mark=at position 5.0cm with {\arrow[black]{stealth}}}, thick]  (p9.south) -- (p15.north);
\draw [postaction=decorate,decoration={markings, mark=at position 4.2cm with {\arrow[black]{stealth}}}, thick]  (p8.south) -- (p15.north);
\draw [postaction=decorate,decoration={markings, mark=at position 1.5cm with {\arrow[black]{stealth}}}, thick]  (p9.south) -- (p13.north);
\draw [postaction=decorate,decoration={markings, mark=at position 1.5cm with {\arrow[black]{stealth}}}, thick]  (p8.south) -- (p12.north);
\end{tikzpicture}}
\caption{Interdependence among the Steps $1$-$10$.}
\label{steps}
\end{figure}

The interdependence among the Steps $1$-$10$ is schematically depicted in \Figref{steps},
and at this stage we have already gathered the main ingredients required to obtain an explicit
expression for the function $S_n$. For instance, in the Step 9 we have obtained a
formula expressing the function $S_n$ as a linear combination of auxiliary
functions $V$ and $W$ defined respectively by (\ref{sv}) and (\ref{rv}). This process
can be thought of as a separation of variables induced by the special zeroes described in Lemma \ref{zero}.
On the other hand, in the Steps 3 and 8 we have shown that the function $V$ satisfy the system of functional
equations (\ref{typeA},\ref{typeD}) under the mappings $L \mapsto L-1$,
$n \mapsto n-1$ and $\mu_i \mapsto \mu_{i+1}$. The same holds for the function $W$ as demonstrated
in the Steps 4 and 7. Thus, since $V$ and $W$ are polynomials of the same order, the linearity
of the system of equations (\ref{typeA}, \ref{typeD}) tells us that $W = \alpha V$ where $\alpha$
is a constant. This uniqueness property employed in our argumentation has been proven in 
\cite{Galleas_2011} under very general conditions.
Moreover, since $S_n$ is essentially a multivariate polynomial and so is the auxiliary function
$V$ due to (\ref{sv}), the residues of $(\Theta_{i,j} + \alpha \Phi_{j,i})/ \mathcal{F}$ must vanish.
This condition tells us that $\alpha = 1$. Thus for $n \leq L$ the formula (\ref{reA}) allows
us to obtain the function $S_n$ up to an overall multiplicative constant starting with the solution
of (\ref{typeA}, \ref{typeD}) for the case $n=1$. The solution $S_1 (\lambda_1^C | \lambda_1^B)$
has been obtained in \Appref{sec:n1} and in what follows we shall demonstrate that the iteration procedure
described by (\ref{reA}) can be mimicked by a multiple contour integral.

\paragraph{Multiple contour integral.} According to the above discussion the expressions
(\ref{reA}) and (\ref{thetaphiA}) can be rewritten as 
\<
\label{omij}
S_n ( X^{1,n} | Y^{1,n} ) = \mathcal{K} \sum_{i,j=1}^{n} \Omega_{i,j} V( X_i^{1,n} | Y_j^{1,n} )
\>
where
\<
\label{KOM}
\mathcal{K} &=& \frac{\prod_{k=1}^{n} a(\lambda_k^C - \mu_1) a(\lambda_k^B - \mu_1)}{\prod_{k=2}^{L} a(\mu_1 - \mu_k) a(\mu_k - \mu_1)} \left[ \prod_{k=1}^{n} \frac{a(\lambda_k^C - \mu_1)}{b(\lambda_k^C - \mu_1)} - \prod_{k=1}^{n} \frac{a(\lambda_k^B - \mu_1)}{b(\lambda_k^B - \mu_1)} \right]^{-1}  \nonumber \\
\Omega_{i,j} &=& \frac{1}{a(\lambda_i^C - \mu_1) b(\lambda_i^C - \mu_1) a(\lambda_j^B - \mu_1) b(\lambda_j^B - \mu_1)} \nonumber \\
&& \times \left[ \prod_{k=1}^{L} a(\lambda_j^B - \mu_k) b(\mu_k - \lambda_i^C) \prod_{\stackrel{k=1}{k \neq i}}^{n} \frac{a(\lambda_i^C - \lambda_k^C)}{b(\lambda_i^C - \lambda_k^C)} \prod_{\stackrel{k=1}{k \neq j}}^{n} \frac{a(\lambda_k^B - \lambda_j^B)}{b(\lambda_k^B - \lambda_j^B)} \right. \nonumber \\
&& \quad - \left. \prod_{k=1}^{L} a(\lambda_i^C - \mu_k) b(\mu_k - \lambda_j^B) \prod_{\stackrel{k=1}{k \neq i}}^{n} \frac{a(\lambda_k^C - \lambda_i^C)}{b(\lambda_k^C - \lambda_i^C)} \prod_{\stackrel{k=1}{k \neq j}}^{n} \frac{a(\lambda_j^B - \lambda_k^B)}{b(\lambda_j^B - \lambda_k^B)} \right] \; .
\>
Moreover, the formula (\ref{omij}) suggests that the function $S_n$ can be expressed as
\<
\label{int} 
S_n ( X^{1,n} | Y^{1,n} ) =  \oint \dots \oint \prod_{i=1}^{n} \frac{\dd w_i}{2 \ii \pi} \frac{\dd \bar{w}_i}{2 \ii \pi} \frac{H(w_1 , \dots , w_n | \bar{w}_1 , \dots , \bar{w}_n)}{ \prod_{i,j=1}^{n} b(w_i - \lambda_j^C) b(\bar{w}_i - \lambda_j^B)} \; ,
\>
where the integrals over the set of variables $\{ w_i \}$ are performed around the contours 
enclosing solely the poles at  $w_i = \lambda_j^C$. On the other hand, the contours associated with
the integration over the set $\{ \bar{w}_i \}$ contain only the poles at $\bar{w}_i = \lambda_j^B$.
The function $H$ is also assumed to be independent of the variables $\lambda_j^{B,C}$ in such a way that 
the only poles contributing to the evaluation of (\ref{int}) are due to the zeroes of 
$\prod_{i,j=1}^{n} b(w_i - \lambda_j^C) b(\bar{w}_i - \lambda_j^B)$.
Now we can integrate the formula (\ref{int}) over the variables $w_1$ and $\bar{w}_1$, and by doing so we
generate the summations over the indexes $i$ and $j$ appearing in (\ref{int}). This procedure allows us to
look for a term by term identification, and the iteration procedure described by (\ref{omij})
is realised if we are able to exhibit a function $H$ satisfying the condition
\<
\label{hdot}
\left. H \right|_{\stackrel{w_1 = \lambda_i^C}{\bar{w}_1 = \lambda_j^B}} &=& \prod_{\stackrel{k=1}{k \neq i}}^{n} b(\lambda_i^C - \lambda_k^C) \prod_{\stackrel{k=1}{k \neq j}}^{n} b(\lambda_j^B - \lambda_k^B) \prod_{k=2}^{n} b(w_k - \lambda_i^C) \prod_{k=2}^{n} b(\bar{w}_k - \lambda_j^B) \nonumber \\
&& \times \; \mathcal{K} \; \Omega_{i,j} \; \bar{H} (w_2 , \dots, w_n | \bar{w}_2 , \dots , \bar{w}_n) \; .
\>
In its turn the function $\bar{H}$ consists of the function $H$, up to an overall multiplicative constant, under the
mappings $L \mapsto L-1$, $n \mapsto n-1$ and $\mu_i \mapsto \mu_{i+1}$ \footnote{Strictly speaking, the relation
(\ref{hdot}) is only required to be valid when integrated as $\oint \dots \oint \left[ \quad  \right] \prod_{i=2}^{n} \dd w_i \dd \bar{w}_i$.}.
Also, it is important to remark here that the relation (\ref{hdot}) needs to be valid for any $i,j \in [1,n]$ and that is required to
hold only when integrated according to (\ref{int}). Thus, under these requirements, we simply need to consider the relation
(\ref{hdot}) under the mappings $\lambda_i^C \mapsto w_1$, $\lambda_j^B \mapsto \bar{w}_1$, $\lambda_k^C \mapsto w_k$ for $k \neq i$ and 
$\lambda_k^B \mapsto \bar{w}_k$ for $k \neq i$ to obtain the relation
\<
\label{REint}
&& H(w_1 , \dots , w_n | \bar{w}_1 , \dots , \bar{w}_n) = \frac{\bar{H}(w_2 ,\dots , w_n | \bar{w}_2 , \dots , \bar{w}_n)}{b(w_1 - \mu_1) b(\bar{w}_1 - \mu_1)} \prod_{k=2}^{n} b(w_1 - w_k)^2  b(\bar{w}_1 - \bar{w}_k)^2 \nonumber \\
&& \times \frac{ \prod_{k=2}^{n} a(w_k - \mu_1) a(\bar{w}_k - \mu_1)}{\prod_{k=2}^{L} a(\mu_1 - \mu_k) a(\mu_k - \mu_1)} \left[ \prod_{k=1}^{n} \frac{a(w_k - \mu_1)}{b(w_k - \mu_1)} - \prod_{k=1}^{n} \frac{a(\bar{w}_k - \mu_1)}{b(\bar{w}_k - \mu_1)} \right]^{-1} \nonumber \\
&& \times \left[ \prod_{k=1}^{L} a(\bar{w}_1 - \mu_k) b(\mu_k - w_1) \prod_{k=2}^{n} \frac{a(w_1 - w_k)}{b(w_1 - w_k)} \prod_{k=2}^{n} \frac{a(\bar{w}_k - \bar{w}_1)}{b(\bar{w}_k - \bar{w}_1)} \right. \nonumber \\
&& \qquad \qquad  \left. - \prod_{k=1}^{L} a(w_1 - \mu_k) b(\mu_k - \bar{w}_1) \prod_{k=2}^{n} \frac{a(w_k - w_1)}{b(w_k - w_1)} \prod_{k=2}^{n} \frac{a(\bar{w}_1 - \bar{w}_k)}{b(\bar{w}_1 - \bar{w}_k)} \right] \; .
\>
Now the expression (\ref{REint}) can be readily iterated once we know the function $H(w_1 | \bar{w}_1)$. 
From formula (\ref{sol1}) we can immediately read that \footnote{Alternative contour integrals expressions are also possible
for the case $n=1$. Here we have chosen (\ref{int1}) in order to have a formula compatible with (\ref{int}).}
\<
\label{int1}
H(w_1 | \bar{w}_1) = c^2 \frac{\left[ \prod_{k=1}^{L} a(w_1 - \mu_k) b(\bar{w}_1 - \mu_k) - \prod_{k=1}^{L} a(\bar{w}_1 - \mu_k) b(w_1 - \mu_k) \right]}{b(w_1 - \mu_1) b(\bar{w}_1 - \mu_1) \left[ \frac{a(w_1 - \mu_1)}{b(w_1 - \mu_1)} - \frac{a(\bar{w}_1 - \mu_1)}{b(\bar{w}_1 - \mu_1)} \right]} \; ,
\>
and the iteration of (\ref{REint}) leave us with the expression,
\<
\label{HH}
&& H(w_1 , \dots , w_n | \bar{w}_1 , \dots , \bar{w}_n) = \nonumber \\
&&  (-1)^{L n + \frac{n(n+1)}{2}} c^{2 n} \frac{\displaystyle \prod_{j > i}^{n} b(w_i - w_j)^2 b(\bar{w}_i - \bar{w}_j)^2 a(w_j - \mu_i) a(\bar{w}_j - \mu_i)}{\prod_{i=1}^{n} b(w_i - \mu_i) b(\bar{w}_i - \mu_i)} \prod_{i=1}^{n} R_i^{-1} \Lambda_i \; , \nonumber \\
\>
where the functions $R_i$ and $\Lambda_i$ are given by
\<
\label{Ri}
R_i &=& \prod_{k=i}^{n} \frac{a(w_k - \mu_i)}{b(w_k - \mu_i)} - \prod_{k=i}^{n} \frac{a(\bar{w}_k - \mu_i)}{b(\bar{w}_k - \mu_i)} \nonumber \\
\Lambda_i &=& \prod_{k=i}^{L} a(\bar{w}_i - \mu_k) b(\mu_k - w_i) \prod_{k=i+1}^{n} \frac{a(w_i - w_k)}{b(w_i - w_k)} \frac{a(\bar{w}_k - \bar{w}_i)}{b(\bar{w}_k - \bar{w}_i)} \nonumber \\ 
&& - \prod_{k=i}^{L} a(w_i - \mu_k) b(\mu_k - \bar{w}_i) \prod_{k=i+1}^{n} \frac{a(w_k - w_i)}{b(w_k - w_i)} \frac{a(\bar{w}_i - \bar{w}_k)}{b(\bar{w}_i - \bar{w}_k)} \; .
\>
The formula (\ref{HH}) already takes into account the asymptotic behaviour described in 
Lemma \ref{asymp}.

\subsection{On-shell formula}
\label{sec:onshell}

In what follows we shall consider the iteration of the relation (\ref{omij}) when
the set of variables $Y^{1,n}$ are constrained by the Bethe ansatz equations (\ref{BA}).
In that case we can see from (\ref{KOM}) that the term $\Omega_{i,j}$ will be the 
most affected and the determination of the function $H(w_1, \dots , w_n | \bar{w}_1, \dots , \bar{w}_n )$
will require some extra effort. For instance, we shall need to promote the function $H$ to
$H^{(s)}(w_s, \dots , w_n | \bar{w}_s, \dots , \bar{w}_n )$ where the index $s$ is introduced
in order to keep track of how many iterations we have performed. The case $s=1$ then yields the
function $H$ entering in formula (\ref{int}). 

We then follow the procedure described in \Secref{sec:offshell} keeping in mind that 
the constraint (\ref{BA}) should hold at each level of the iteration process. 
By doing so we find the relation,
\<
\label{Hs}
&& H^{(s)}(w_s, \dots , w_n | \bar{w}_s, \dots , \bar{w}_n ) = \frac{(-1)^{L} \prod_{k=1}^{L} b(\bar{w}_s - \mu_k) \prod_{k=s+1}^{n} a(w_k - \mu_s) b(w_k - w_s)}{b(w_s - \mu_s) b(\bar{w}_s - \mu_s) \prod_{k=s+1}^{L} a(\mu_s - \mu_k) a(\mu_k - \mu_s)} \nonumber \\
&& \quad \times \prod_{k=s+1}^{n} a(\bar{w}_s - \bar{w}_k) a(\bar{w}_k - \mu_s) b(\bar{w}_k - \bar{w}_s) \left[ \prod_{k=s}^{n} \frac{a(w_k - \mu_s)}{b(w_k - \mu_s)} - \prod_{k=s}^{n} \frac{a(\bar{w}_k - \mu_s)}{b(\bar{w}_k - \mu_s)} \right]^{-1} \nonumber \\
&& \quad \times \left[ \frac{(-1)^n}{\prod_{k=1}^{s-1} b(\bar{w}_s - \mu_k)} \prod_{k=s}^{L} a(w_s - \mu_k) \prod_{k=s+1}^{n} a(w_k - w_s) \right. \nonumber \\
&& \qquad \qquad  \left. + \frac{\phi_2 \phi_1^{-1}}{\prod_{k=1}^{s-1} a(\bar{w}_s - \mu_k)} \prod_{k=s}^{L} b(w_s - \mu_k) \prod_{k=s+1}^{n} a(w_s - w_k) \prod_{k=1}^{s-1} \frac{a(\bar{w}_s - \bar{w}_k)}{a(\bar{w}_k - \bar{w}_s)} \right] \nonumber \\
&& \quad \times  H^{(s+1)}(w_{s+1}, \dots , w_n | \bar{w}_{s+1}, \dots , \bar{w}_n ) \; ,
\>
which needs to be iterated starting from $s=1$ up to $s=n-1$. Then at the last step we shall also need the
function $H^{(n)}(w_n | \bar{w}_n )$ which is readily obtained from (\ref{sol1}) taking into
account the relation (\ref{BA}). Thus we have
\<
\label{Hn}
&& H^{(n)}(w_n | \bar{w}_n ) = (-1)^n c^2 \frac{\prod_{k=1}^{L} b(\bar{w}_n - \mu_k)}{b(w_n - \mu_n) b(\bar{w}_n - \mu_n)} \left[ \frac{a(w_n - \mu_n)}{b(w_n - \mu_n)} - \frac{a(\bar{w}_n - \mu_n)}{b(\bar{w}_n - \mu_n)} \right]^{-1} \nonumber \\
&& \quad \times \left[ \frac{(-1)^n}{\prod_{k=1}^{n-1} b(\bar{w}_n - \mu_k)} \prod_{k=n}^{L} a(w_n - \mu_k) \right. \nonumber \\
&& \qquad \qquad \quad \left. + \frac{\phi_2 \phi_1^{-1}}{\prod_{k=1}^{n-1} a(\bar{w}_n - \mu_k)} \prod_{k=n}^{L} b(w_n - \mu_k)  \prod_{k=1}^{n-1} \frac{a(\bar{w}_n - \bar{w}_k)}{a(\bar{w}_k - \bar{w}_n)} \right] \; ,  \nonumber \\
\>
and the iteration process above described yields the formula
\<
\label{FON}
&& H(w_1, \dots , w_n | \bar{w}_1, \dots , \bar{w}_n ) = \nonumber \\
&& \qquad (-1)^{(n-1)(L + \frac{n}{2})} c^{2n} \frac{\prod_{i=1}^{n} \prod_{j=1}^{L} b(\bar{w}_i - \mu_j) \prod_{\stackrel{i,j=1}{j>i}}^{n} a(\bar{w}_i - \bar{w}_j) a(\bar{w}_j - \mu_i) b(\bar{w}_j - \bar{w}_i) }{\prod_{i=1}^{n} b(\bar{w}_i - \mu_i)} \nonumber \\
&& \qquad \times  \frac{\prod_{\stackrel{i,j=1}{j>i}}^{n} a(w_j - \mu_i) b(w_j - w_i)}{\prod_{i=1}^{n} b(w_i - \mu_i)} \prod_{i=1}^{n} R_i^{-1} \Lambda_i^{\mbox{{\tiny ON}}} \; ,
\>
where $R_i$ is given in (\ref{Ri}) and

\<
\label{LON}
\Lambda_i^{\mbox{{\tiny ON}}} &=& \frac{(-1)^n}{\prod_{k=1}^{i-1} b(\bar{w}_i - \mu_k)} \prod_{k=i}^{L} a(w_i - \mu_k) \prod_{k=i+1}^{n} a(w_k - w_i) \nonumber \\
&& + \frac{\phi_2 \phi_1^{-1}}{\prod_{k=1}^{i-1} a(\bar{w}_i - \mu_k)} \prod_{k=i}^{L} b(w_i - \mu_k) \prod_{k=i+1}^{n} a(w_i - w_k) \prod_{k=1}^{i-1} \frac{a(\bar{w}_i - \bar{w}_k)}{a(\bar{w}_k - \bar{w}_i)} \; . \nonumber \\
\>

The formulas (\ref{int}) and (\ref{FON}) constitute an expression for the scalar product $S_n$ 
under the constraint (\ref{BA}). In that case we have available Slavnov's formula \cite{Slavnov_1989}
and (\ref{int}, \ref{FON}) then corresponds to an alternative representation. Here it is 
also important to remark that it might still be possible to obtain simpler representations for the 
on-shell scalar product from the functional equations (\ref{typeA}) and (\ref{typeD}). For instance, notice that
the quantity $ \phi_1 N_i^{(B)} + \phi_2 \widetilde{N}_i^{(B)}$ vanishes under the 
condition (\ref{BA}). Thus if we multiply Eq. (\ref{typeA}) by $\phi_1$ and add it to Eq. (\ref{typeD}) multiplied
by $\phi_2$, we are left with the relation
\<
\label{onfun}
K_0 \; S_n (X^{1,n}) + \sum_{i=1}^{n} K_i S_n (X_{i}^{0,n}) = 0 
\>
with coefficients $K_0 = \phi_1 M_0 + \phi_2 \widetilde{M}_0$ and $K_i = \phi_1 N_i^{(C)} + \phi_2 \widetilde{N}_i^{(C)}$ for 
$i \in [1,n]$. In Eq. (\ref{onfun}) we have omitted the dependence of $S_n$ with the set of variables
$Y^{1,n}$ since now they are completely fixed by the relation (\ref{BA}). Moreover, we have checked for small values
of $L$ and $n$ that Eq. (\ref{onfun}) solely is able to determine the on-shell scalar product
up to an overall multiplicative factor. Thus the direct study of (\ref{onfun}) might still offer the possibility
of deriving alternative representations.

\section{Concluding remarks}
\label{sec:conclusion}

In this work we have derived functional relations describing the scalar product
of Bethe vectors for the six-vertex model. The origin of the functional equations
is a consistency condition between the Yang-Baxter algebra and the
highest weight representation theory of the $\alg{sl}_2$ algebra. More
precisely, we have obtained a system formed by two equations (\ref{typeA}, \ref{typeD})
whose resolution then yields a multiple integral representation (\ref{int}, \ref{HH})
for the aforementioned scalar product. Although a multiple integral formula for this
same scalar product had been obtained previously in \cite{deGier_Galleas11} as a specialisation
of Baxter's solution of the Z-invariant six-vertex model, the structure of the formula
obtained here is considerably different.

The resolution of our system of functional equations follows a simple sequence 
of systematic steps which are described in \Secref{sec:offshell}. It is worth remarking
that the method employed in \Secref{sec:offshell} consists of an extension of the
one developed in \cite{Galleas_2011, Galleas_2012} for the partition function of SOS
models with domain wall boundaries.
 
Moreover, in \Secref{sec:onshell} we have also obtained a formula for the scalar product
(\ref{scp}) under the on-shell condition (\ref{BA}). In that case it is well known that
the six-vertex model scalar products are given by Slavnov's determinant formula \cite{Slavnov_1989},
and thus the formula (\ref{int}, \ref{FON}) constitutes an alternative representation.
Furthermore, here we have also obtained a single functional equation describing  
scalar products under the on-shell condition (\ref{onfun}). In this way the direct study
of (\ref{onfun}) might still offer the possibility of deriving different representations.

Recently, there has been a lot of discussion on the possibility of obtaining determinant
representations for on-shell scalar products in models based on higher rank algebras 
\cite{Wheeler_2012, Belliard1, Belliard2, Belliard3, Belliard4}. In particular, some doubts on that possibility have
been put forward in \cite{Belliard1}. Since the method described here is based on the
Yang-Baxter algebra, which is the common algebraic structure underlying integrable vertex 
models, it would be interesting to investigate the extension of this method for higher
rank algebras together with the generalisation of the integral formulas (\ref{int}, \ref{HH} , \ref{FON}).

\section{Acknowledgements}
\label{sec:ack}
The author is supported by the Netherlands Organisation for Scientific
Research (NWO) under the VICI grant 680-47-602. The work of W. Galleas is also
part of the ERC Advanced grant research programme No. 246974, 
{\it ``Supersymmetry: a window to non-perturbative physics"}.

\appendix

\section{Polynomial structure}
\label{sec:POL}

The structure of the function $S_n$ with a given variable $\lambda_i^B$ is
completely encoded in the operator $B(\lambda_i^B)$, while its dependence with
the variable $\lambda_j^C$ is described by the operator $C(\lambda_j^C)$. 
Thus we mainly need to characterise the dependence of the operators $B(\lambda)$
and $C(\lambda)$ with its spectral parameter in order to describe the structure
of $S_n$ with respect to the set of variables $\{ \lambda_i^B \}$ and $\{ \lambda_i^C \}$.
In order to proceed with this analysis it will be useful to recast formula
(\ref{abcd}) as
\<
\label{abcdL}
\mathcal{T}^{(L)} (\lambda) = \left( \begin{matrix}
A_L(\lambda) & B_L(\lambda) \cr
C_L(\lambda) & D_L(\lambda) \end{matrix} \right) \; ,
\>
where the have introduced the index $L$ to emphasise we are considering the ordered
product of $L$ matrices $\mathcal{R}_{a j}$ according to (\ref{rep}).
In its turn the matrix $\mathcal{R}_{a j}$ consists of a matrix in the space
$\mathbb{V}_a$ whose entries are then matrices acting non-trivially on the
$j$-space of the tensor product $\mathbb{V}_1 \otimes \dots \otimes \mathbb{V}_L$.
Here $\mathbb{V}_j \cong \mathbb{C}^2$ and more precisely we have
\<
\label{raj}
\mathcal{R}_{a j}  = \left( \begin{matrix}
\alpha_j  & \beta_j  \nonumber \\
\gamma_j  & \delta_j \end{matrix} \right) 
\>
where
\begin{align}
\label{alfa}
\alpha_j   &= \left( \begin{matrix}
a & 0 \cr
0 & b \end{matrix} \right)_j &
\beta_j &= \left( \begin{matrix}
0 & 0 \cr
c & 0  \end{matrix} \right)_j \nonumber \\
\gamma_j &= \left( \begin{matrix}
0 & c \cr
0 & 0  \end{matrix} \right)_j &
\delta_j &= \left( \begin{matrix}
b & 0 \cr
0 & a \end{matrix} \right)_j \; .
\end{align}
Thus the construction of the operators $A$, $B$, $C$ and $D$ described in
(\ref{rep}) can be implemented in a recursive manner with the help of the relation
\<
\label{recursion}
\mathcal{T}^{(L+1)} (\lambda) = \mathcal{T}^{(L)} (\lambda) \; \mathcal{R}_{a L+1}(\lambda - \mu_{L+1}) 
\>
and initial conditions
\begin{align}
A_1 (\lambda) = & \alpha_1 (\lambda - \mu_1) & B_1 (\lambda) & = \beta_1 (\lambda - \mu_1) \nonumber \\
C_1 (\lambda) = & \gamma_1 (\lambda - \mu_1) & D_1 (\lambda) & = \delta_1 (\lambda - \mu_1) \; . 
\end{align}
In terms of its components the recursion relation (\ref{recursion}) reads
\<
\label{reABCD}
A_{L+1} (\lambda) &=& A_{L} (\lambda) \alpha_{L+1} (\lambda - \mu_{L+1}) + B_{L} (\lambda) \gamma_{L+1} (\lambda - \mu_{L+1}) \nonumber \\
B_{L+1} (\lambda) &=& A_{L} (\lambda) \beta_{L+1} (\lambda - \mu_{L+1}) + B_{L} (\lambda) \delta_{L+1} (\lambda - \mu_{L+1}) \nonumber \\
C_{L+1} (\lambda) &=& C_{L} (\lambda) \alpha_{L+1} (\lambda - \mu_{L+1}) + D_{L} (\lambda) \gamma_{L+1} (\lambda - \mu_{L+1}) \nonumber \\
D_{L+1} (\lambda) &=& C_{L} (\lambda) \beta_{L+1} (\lambda - \mu_{L+1}) + D_{L} (\lambda) \delta_{L+1} (\lambda - \mu_{L+1}) \; ,
\> 
which allows us to infer the structure of the operators $A(\lambda)$, $B(\lambda)$, $C(\lambda)$ and $D(\lambda)$
with respect to the spectral parameter $\lambda$. For that it is important to keep in mind that the spectral
parameter $\lambda$ only enters those operators through the functions $a$ and $b$ since $c$ depends only on the
parameter $\gamma$.
In terms of the variables $x = e^{2 \lambda}$ and $q = e^{\gamma}$ we have $a(x) = 2^{-1} x^{-\frac{1}{2}} (x q - q^{-1})$,
$b(x) = 2^{-1} x^{-\frac{1}{2}} (x  - 1)$ and $c(x) = 2^{-1} (q - q^{-1})$, and in this way the entries of
$\mathcal{R}_{aj}$ can be written in the form
\begin{align}
\label{alfafa}
\alpha_j (x) = & 2^{-1} x^{-\frac{1}{2}} P_{\alpha}(x)      & \beta_j (x) & =  Q_{\beta} \nonumber \\
\gamma_j (x) = & Q_{\gamma}  & \delta_j (x) & = 2^{-1} x^{-\frac{1}{2}} P_{\delta} (x) \; .
\end{align}
The terms $P_{\alpha}(x)$ and $P_{\delta}(x)$ in (\ref{alfafa}) are polynomials of order $1$, while 
$Q_{\beta}$ and $Q_{\gamma}$ are constants. Under these considerations the iteration of (\ref{reABCD})
tells us that the operators $B(\lambda)$ and $C(\lambda)$ are of the form,
\<
\label{BC}
B(\lambda) = x^{- \frac{(L-1)}{2}} P_{B} (x) \qquad \qquad C(\lambda) = x^{- \frac{(L-1)}{2}} P_{C} (x) \;
\>
where $P_B (x)$ and $P_C (x)$ are both polynomials of order $L-1$.
Thus the definition (\ref{scp}) combined with the property (\ref{BC}) immediately tell
us that the scalar product $S_n$ is of the form
\<
S_n  = \prod_{i=1}^{n} (x_i^B x_i^C )^{-\frac{L-1}{2}} \bar{S}_n (x_1^C , \dots , x_n^C | x_1^B , \dots , x_n^B ) \; ,
\>
where $x_i^{B,C} = e^{2 \lambda_i^{B,C}}$ and $\bar{S}_n$ is a polynomial of order $L-1$ in each one of its variables.

\section{Special zeroes}
\label{sec:ZERO}

The method considered in \Secref{sec:offshell} for solving the system of Eqs.
(\ref{typeA}, \ref{typeD}) relies on the identification of certain
zeroes of $S_n$ as function of the set of variables $\{ \lambda_i^{B} \}$ and $\{ \lambda_i^{C} \}$.
This is due to the fact that we are interested in a polynomial solution, and as such it can be
characterised by its zeroes. In what follows we shall proceed with the identification of those zeroes for particular
values of $n$ for illustrative purposes, and subsequently consider the general case.

\subsection{Case $n=2$}
The coefficients $N_1^{(C)}$ and $N_2^{(C)}$ vanish for the specialisation $\lambda_1^C = \mu_1$
and $\lambda_2^C = \mu_1 - \gamma$. Thus for this particular choice of variables $\lambda_{1,2}^C$,
Eq. (\ref{typeA}) simplifies to
\<
\label{N2a}
M_0 S_2 (\mu_1 , \mu_1 - \gamma | \lambda_1^B , \lambda_2^B) + N_1^{B} S_2 (\mu_1 , \mu_1 - \gamma | \lambda_0 , \lambda_2^B)
+ N_2^{B} S_2 (\mu_1 , \mu_1 - \gamma | \lambda_0 , \lambda_1^B) = 0 \; . \nonumber \\
\>
Next we set $\lambda_0 = \mu_1 - \gamma$ noticing that Eq. (\ref{N2a}) does not simplify significantly.
Nevertheless, the function $M_0$ then consists of a single term. In this way we obtain the following relation,
\<
\label{N2b}
S_2 (\mu_1 , \mu_1 - \gamma | \lambda_1^B , \lambda_2^B) =  &-& \frac{N_1^{(B)}}{M_0} S_2 (\mu_1 , \mu_1 - \gamma | \mu_1 - \gamma , \lambda_2^B) \nonumber \\
&-& \frac{N_2^{(B)}}{M_0} S_2 (\mu_1 , \mu_1 - \gamma | \mu_1 - \gamma , \lambda_1^B) \; . 
\>
Eq. (\ref{typeD}) also simplifies similarly under the specialisation $\lambda_1^C = \mu_1$ and $\lambda_2^C = \mu_1 - \gamma$.
In that case it reads, 
\<
\label{N2c}
\widetilde{M}_0 S_2 (\mu_1 , \mu_1 - \gamma | \lambda_1^B , \lambda_2^B) + \widetilde{N}_1^{B} S_2 (\mu_1 , \mu_1 - \gamma | \lambda_0 , \lambda_2^B)
+ \widetilde{N}_2^{B} S_2 (\mu_1 , \mu_1 - \gamma | \lambda_0 , \lambda_1^B) = 0 \; , \nonumber \\
\>
and we can substitute (\ref{N2b}) into (\ref{N2c}) to obtain 
\<
\label{N2d}
&& Q_0 S_2 (\mu_1 , \mu_1 - \gamma | \mu_1 - \gamma , \lambda_0) \nonumber \\
&& + Q_1 S_2 (\mu_1 , \mu_1 - \gamma | \mu_1 - \gamma , \lambda_1^B) + Q_2 S_2 (\mu_1 , \mu_1 - \gamma | \mu_1 - \gamma , \lambda_2^B) = 0 \; .
\>
In terms of variables $x_i^B = e^{2 \lambda_i^B}$, the functions $Q_i$ are rational functions whose explicit forms
are not enlightening. Nevertheless, they possess very non-trivial zeroes and by setting $\lambda_i^B = \{ r_i \; | \; Q_i (r_i) = 0 \}$
for $i=1,2$ we are left with the relation
\<
\label{N2e}
\left. Q_0 \right|_{\stackrel{\lambda_1^B = r_1}{\lambda_2^B = r_2}} S_2 (\mu_1 , \mu_1 - \gamma | \mu_1 - \gamma , \lambda_0) = 0 \; .
\>
Thus we can conclude that $S_2 (\mu_1 , \mu_1 - \gamma | \mu_1 - \gamma , \lambda_0)$ vanishes and by taking this result
back to (\ref{N2b}) we obtain $S_2 (\mu_1 , \mu_1 - \gamma | \lambda_1^B , \lambda_2^B) = 0$. The same procedure
can be employed considering initially the variables $\lambda_i^B$ instead of $\lambda_i^C$, and by doing so
we end up with the property $S_2 ( \lambda_1^C , \lambda_2^C | \mu_1 , \mu_1 - \gamma ) = 0$.

\subsection{Case $n=3$}

We set $\lambda_1^C = \mu_1$ and $\lambda_2^C = \mu_1 - \gamma$ in Eq. (\ref{typeA}) such that
$N_1^{(C)} = N_2^{(C)} = 0$. In that case (\ref{typeA}) simplifies to
\<
\label{N3a}
&& M_0 S_3 (\mu_1 , \mu_1 - \gamma, \lambda_3^C | \lambda_1^B , \lambda_2^B , \lambda_3^B) + N_3^{(C)} S_3 (\lambda_0 , \mu_1 , \mu_1 - \gamma | \lambda_1^B , \lambda_2^B , \lambda_3^B)  \nonumber \\
&& + N_1^{(B)} S_3 (\mu_1 , \mu_1 - \gamma, \lambda_3^C | \lambda_0 , \lambda_2^B , \lambda_3^B) + N_2^{(B)} S_3 (\mu_1 , \mu_1 - \gamma, \lambda_3^C | \lambda_0 , \lambda_1^B , \lambda_3^B) \nonumber \\
&& + N_3^{(B)} S_3 (\mu_1 , \mu_1 - \gamma, \lambda_3^C | \lambda_0 , \lambda_1^B , \lambda_2^B ) = 0 \; . 
\>
Next we set $\lambda_1^B = \mu_1 - \gamma$ and $\lambda_2^B = \mu_1$ in (\ref{N3a}) and obtain the relation
\<
\label{N3b}
&& M_0 S_3 (\mu_1 , \mu_1 - \gamma, \lambda_3^C | \mu_1 - \gamma , \mu_1 , \lambda_3^B) + N_3^{(C)} S_3 (\lambda_0 , \mu_1 , \mu_1 - \gamma | \mu_1 - \gamma , \mu_1 , \lambda_3^B)  \nonumber \\
&& + N_3^{(B)} S_3 (\mu_1 , \mu_1 - \gamma, \lambda_3^C | \lambda_0 , \mu_1 - \gamma , \mu_1 ) = 0 \; . 
\> 
In addition to that we set $\lambda_0 = \mu_1 - \gamma$ in (\ref{N3b}) which then yields the following relation,
\<
\label{N3c}
&& M_0 S_3 (\mu_1 , \mu_1 - \gamma, \lambda_3^C | \mu_1 - \gamma , \mu_1 , \lambda_3^B) + N_3^{(C)} S_3 (\mu_1 - \gamma , \mu_1 , \mu_1 - \gamma | \mu_1 - \gamma , \mu_1 , \lambda_3^B)  \nonumber \\
&& + N_3^{(B)} S_3 (\mu_1 , \mu_1 - \gamma, \lambda_3^C | \mu_1 - \gamma , \mu_1 - \gamma , \mu_1 ) = 0 \; . 
\> 
Furthermore, the relation (\ref{N3c}) simplifies to
\<
\label{N3d}
S_3 (\mu_1 - \gamma , \mu_1 , \mu_1 - \gamma | \mu_1 - \gamma , \mu_1 , \lambda) = S_3 (\mu_1 , \mu_1 - \gamma, \lambda | \mu_1 - \gamma , \mu_1 - \gamma , \mu_1 ) \; ,
\>
for $\lambda_3^B = \lambda_3^C = \lambda$ since the coefficient $M_0$ is proportional to $b(\lambda_3^C - \lambda_3^B)$.
Then we can see that Eq. (\ref{N3c}) combined with the property (\ref{N3d}) yields the relation
\<
\label{N3e}
S_3 (\mu_1 , \mu_1 - \gamma, \lambda_3^C | \mu_1 - \gamma , \mu_1 , \lambda_3^B) = &-& \frac{N_3^{(C)}}{M_0} S_3 (\mu_1 - \gamma , \mu_1 , \mu_1 - \gamma | \mu_1 - \gamma , \mu_1 , \lambda_3^B) \nonumber \\
&-& \frac{N_3^{(B)}}{M_0} S_3 (\mu_1 - \gamma , \mu_1 , \mu_1 - \gamma | \mu_1 - \gamma , \mu_1 , \lambda_3^C) \; , \nonumber \\
\>
which can be substituted back into (\ref{N3b}) considering the symmetry property (\ref{symm}).  We are then left with an
equation of the form 
\<
\label{N3f}
&& Q_0 S_3 (\mu_1 - \gamma , \mu_1 , \mu_1 - \gamma | \mu_1 - \gamma , \mu_1 , \lambda_0 ) + Q_3^{B} S_3 (\mu_1 - \gamma , \mu_1 , \mu_1 - \gamma | \mu_1 - \gamma , \mu_1 , \lambda_3^{B} ) \nonumber \\
&& + Q_3^{C} S_3 (\mu_1 - \gamma , \mu_1 , \mu_1 - \gamma | \mu_1 - \gamma , \mu_1 , \lambda_3^{C} ) = 0 \; ,
\>
where $Q_0$ and $Q_3^{B,C}$ are rational functions in the variables $x_i^{B,C}$. Then we can set
$\lambda_3^{B,C} = \{ r^{B,C} \; | \; Q_3^{B,C}(r^{B,C}) = 0 \}$ in (\ref{N3f}) and conclude that
$S_3 (\mu_1 - \gamma , \mu_1 , \mu_1 - \gamma | \mu_1 - \gamma , \mu_1 , \lambda_0 ) = 0$. Considering (\ref{N3e}),
the latter also implies in 
\<
\label{N3g}
S_3 (\mu_1 , \mu_1 - \gamma, \lambda_3^C | \mu_1 - \gamma , \mu_1 , \lambda_3^B) = 0 \; .
\>
Next we set $\lambda_0 = \mu_1$ and $\lambda_1^B = \mu_1 - \gamma$ in Eq. (\ref{typeA})
considering the property (\ref{N3g}). By doing so we obtain the following relation,
\<
\label{N3h}
M_0 S_3 (\mu_1 , \mu_1 - \gamma , \lambda_3^C | \mu_1 - \gamma , \lambda_2^B , \lambda_3^B) + N_3^{(C)} S_3 (\mu_1 , \mu_1 , \mu_1 - \gamma | \mu_1 - \gamma , \lambda_2^B , \lambda_3^B) = 0 \; . \nonumber \\
\>
The coefficient $M_0$ in (\ref{N3h}) vanishes for $\lambda_3^C = \mu_1 - \gamma$ while $N_3^{(C)}$ is finite. Thus 
for this particular specialisation of $\lambda_3^C$ we obtain 
$S_3 (\mu_1 , \mu_1 , \mu_1 - \gamma | \mu_1 - \gamma , \lambda_2^B , \lambda_3^B) = 0$, and consequently we can conclude that
\<
\label{N3i}
S_3 (\mu_1 , \mu_1 - \gamma , \lambda_3^C | \mu_1 - \gamma , \lambda_2^B , \lambda_3^B) = 0 \; .
\>

So far we have only considered Eq. (\ref{typeA}) and the next step is to set $\lambda_1^C = \mu_1$,
$\lambda_2^C = \mu_1 - \gamma$ and $\lambda_0 = \mu_1 - \gamma$ in Eq. (\ref{typeD}). We also take into 
account the property (\ref{N3i}) and by doing so we obtain the relation
\<
\label{N3j}
\widetilde{M}_0 S_3 (\mu_1 , \mu_1 - \gamma , \lambda_3^C | \lambda_1^B , \lambda_2^B , \lambda_3^B ) + \widetilde{N}_3^{(C)} S_3 (\mu_1 - \gamma , \mu_1 , \mu_1 - \gamma | \lambda_1^B , \lambda_2^B , \lambda_3^B ) = 0 \; .
\>
The function $\widetilde{M}_0$ possesses non-trivial zeroes and we can set $\lambda_3^C$ in (\ref{N3j})
such that $\widetilde{M}_0$ vanishes in order to conclude that $S_3 (\mu_1 - \gamma , \mu_1 , \mu_1 - \gamma | \lambda_1^B , \lambda_2^B , \lambda_3^B ) = 0$. 
Then from (\ref{N3j}) the latter implies in
\<
\label{N3k}
S_3 (\mu_1 , \mu_1 - \gamma , \lambda_3^C | \lambda_1^B , \lambda_2^B , \lambda_3^B) = 0 \; .
\>

The procedure above described can also be performed with the specialisations of the variables
$\lambda_i^B$ and $\lambda_i^C$ exchanged allowing us to conclude that 
\<
\label{N3l}
S_3 (\lambda_1^C , \lambda_2^C , \lambda_3^C | \mu_1 , \mu_1 - \gamma , \lambda_3^B ) = 0 \; .
\>

\subsection{General case}

We consider Eq. (\ref{typeA}) under the specialisations $\lambda_1^C = \lambda_2^B = \mu_1$,
$\lambda_2^C = \lambda_1^B = \mu_1 - \gamma$ and $\lambda_j^{B,C} = \lambda_{j+1}^{B,C} + \gamma$
for $j \in [3,n-1]$ and collect the results at each step. At the final step the only non-vanishing coefficients
are $M_0$ and $N_n^{(B,C)}$, and we then set $\lambda_0 = \mu_1 - \gamma$. This procedure then yields the formula,
\<
\label{Nna}
S_n (X | Y)  = -\frac{N_n^{(C)}}{M_0} S_n (X^{*} | Y) -\frac{N_n^{(B)}}{M_0} S_n (X | Y^{*}) 
\>
where
\<
\label{Nnb}
X &=& \{ \mu_1 , \mu_1 - \gamma , \lambda_n^C + (n-3)\gamma, \dots , \lambda_n^C + (n-j)\gamma , \dots , \lambda_n^C   \} \nonumber \\
Y &=& \{ \mu_1 - \gamma , \mu_1 , \lambda_n^B + (n-3)\gamma, \dots , \lambda_n^B + (n-j)\gamma , \dots , \lambda_n^B \} \nonumber \\
X^{*} &=& \{ \mu_1 - \gamma , \mu_1 , \mu_1 - \gamma , \lambda_n^C + (n-3)\gamma, \dots , \lambda_n^C + (n-j)\gamma , \dots , \lambda_n^C + \gamma  \} \nonumber \\
Y^{*} &=& \{ \mu_1 - \gamma , \mu_1 - \gamma , \mu_1 , \lambda_n^B + (n-3)\gamma, \dots , \lambda_n^B + (n-j)\gamma , \dots , \lambda_n^B + \gamma \} \; . \nonumber \\
\>
The relation (\ref{Nna}) can now be substituted back into the previous steps leading to it,
and an analysis similar to the one employed for the cases $n=2,3$ allows us to conclude that
\<
\label{Nnc}
S_n (\mu_1 , \mu_1 - \gamma, \lambda_3^C , \dots , \lambda_n^C | \mu_1 - \gamma, \lambda_2^B , \dots , \lambda_n^B ) = 0 \; .
\>
Next we set $\lambda_1^C = \mu_1$, $\lambda_2^C = \mu_1 - \gamma$ and $\lambda_0 = \mu_1 - \gamma$
in Eq. (\ref{typeD}) taking into account the property (\ref{Nnc}). By doing so we obtain the more
general condition 
\<
\label{Nnd}
S_n (\mu_1 , \mu_1 - \gamma, \lambda_3^C , \dots , \lambda_n^C | \lambda_1^B,  \dots , \lambda_n^B ) = 0 \; .
\>

The procedure above described can also be performed with the specialisations of the variables
$\lambda_i^B$ and $\lambda_i^C$ exchanged. In that case we then obtain the vanishing condition
\<
\label{Nne}
S_n (\lambda_1^C,  \dots , \lambda_n^C | \mu_1 , \mu_1 - \gamma, \lambda_3^B , \dots , \lambda_n^B  ) = 0 \; .
\>

\section{$S_n$ as a doubly symmetric function}
\label{sec:sym}

The function $S_n (\lambda_1^{C} , \dots , \lambda_n^{C} | \lambda_1^{B} , \dots , \lambda_n^{B})$
defined by (\ref{scp}) is expected to be invariant under the exchange of variables  
$\lambda_i^{B} \leftrightarrow \lambda_j^{B}$ and $\lambda_i^{C} \leftrightarrow \lambda_j^{C}$
separately. This property is due to the definition (\ref{scp}) and the commutation
relations
\[
\left[ B(\lambda_1) , B(\lambda_2) \right] = \left[ C(\lambda_1) , C(\lambda_2) \right] = 0
\]
described in (\ref{commutAB}) and (\ref{commutAC}).
Nevertheless, once we assume $S_n$ is determined by the Eqs. (\ref{typeA}) and (\ref{typeD}),
it would be desirable this symmetry to be an inherent property of their solutions. This is indeed the
case and it can be demonstrated as follows.

We firstly notice that the coefficients $M_0$ and $N_i^{(B)}$ defined in (\ref{coeffA})
are the only coefficients exhibiting poles when $\lambda_0 \rightarrow \lambda_i^{B}$. 
Moreover, those coefficients also satisfy the property
\<
\label{pp}
\lim_{\lambda_0 \rightarrow \lambda_i^{B}} M_0 \; b(\lambda_0 - \lambda_i^{B}) &=& - \lim_{\lambda_0 \rightarrow \lambda_i^{B}} N_i^{B} \; b(\lambda_0 - \lambda_i^{B}) \nonumber \\
&=& c \prod_{j=1}^{L} a(\lambda_i^{B} - \mu_j) \prod_{\stackrel{j=1}{j \neq i}}^{n} \frac{a(\lambda_j^{B} - \lambda_i^{B})}{b(\lambda_j^{B} - \lambda_i^{B})} \; .
\>
Next we integrate Eq. (\ref{typeA}) over the variable $\lambda_0$ and around the contour
$\mathcal{C}_i^{B}$ enclosing solely the variable $\lambda_i^{B}$. Thus considering the property (\ref{pp})
we obtain the identity
\<
\label{almostB}
&& c \prod_{j=1}^{L} a(\lambda_i^{B} - \mu_j) \prod_{\stackrel{j=1}{j \neq i}}^{n} \frac{a(\lambda_j^{B} - \lambda_i^{B})}{b(\lambda_j^{B} - \lambda_i^{B})} \left[ S_n (\lambda_1^{C} , \dots , \lambda_n^{C} | \lambda_1^{B} , \dots , \lambda_n^{B} ) \right. \nonumber \\
&& \qquad \qquad\qquad \qquad \qquad - \left. S_n (\lambda_1^{C} , \dots , \lambda_n^{C} | \lambda_i^{B} , \lambda_1^{B} , \dots , \lambda_{i-1}^{B} , \lambda_{i+1}^{B} , \dots  , \lambda_n^{B} ) \right] = 0 \; . \nonumber \\
\>
The formula (\ref{almostB}) allows us to conclude that
\<
S_n (\lambda_1^{C} , \dots , \lambda_n^{C} | \lambda_i^{B} , \lambda_1^{B} ,\dots , \lambda_{i-1}^{B} , \lambda_{i+1}^{B} , \dots  , \lambda_n^{B} ) = S_n (\lambda_1^{C} , \dots , \lambda_n^{C} | \lambda_1^{B} , \dots , \lambda_n^{B} ) \nonumber \\
\>
for $i \in [1, L]$, which implies the symmetry relation
\[
\label{sym1}
S_n ( \lambda_1^{C} , \dots , \lambda_n^{C} | \dots , \lambda_i^{B} , \dots , \lambda_j^{B} , \dots ) = S_n ( \lambda_1^{C} , \dots , \lambda_n^{C} | \dots , \lambda_j^{B} , \dots , \lambda_i^{B} , \dots ) \; .
\]

Now in order to demonstrate an equivalent symmetry relation with respect to the set of variables
$\{ \lambda_i^C \}$, we integrate Eq. (\ref{typeA}) over the variable $\lambda_0$
and around the contour $\mathcal{C}_i^{C}$ containing solely the variable $\lambda_i^{C}$.
We also take into account that the coefficients $M_0$ and $N_i^{(C)}$ are the only ones
possessing poles when $\lambda_0 \rightarrow \lambda_i^{C}$.
In this way we obtain the identity
\<
\label{almostC}
&& c \prod_{j=1}^{L} a(\lambda_i^{C} - \mu_j) \prod_{\stackrel{j=1}{j \neq i}}^{n} \frac{a(\lambda_j^{C} - \lambda_i^{C})}{b(\lambda_j^{C} - \lambda_i^{C})} \left[ S_n (\lambda_1^{C} , \dots , \lambda_n^{C} | \lambda_1^{B} , \dots , \lambda_n^{B} ) \right. \nonumber \\
&& \qquad \qquad \qquad \qquad \qquad - \left. S_n ( \lambda_i^{C} , \lambda_1^{C} , \dots , \lambda_{i-1}^{C} , \lambda_{i+1}^{C} , \dots  , \lambda_n^{C} | \lambda_1^{B} , \dots , \lambda_n^{B} ) \right] = 0 \; , \nonumber \\
\>
with the help of the property
\<
\lim_{\lambda_0 \rightarrow \lambda_i^{C}} M_0 \; b(\lambda_0 - \lambda_i^{C}) &=& - \lim_{\lambda_0 \rightarrow \lambda_i^{C}} N_i^{(C)} \; b(\lambda_0 - \lambda_i^{C}) \nonumber \\
&=& - c \prod_{j=1}^{L} a(\lambda_i^{C} - \mu_j) \prod_{\stackrel{j=1}{j \neq i}}^{n} \frac{a(\lambda_j^{C} - \lambda_i^{C})}{b(\lambda_j^{C} - \lambda_i^{C})} \; .
\>
The property $S_n ( \lambda_i^{C} , \lambda_1^{C} ,\dots , \lambda_{i-1}^{C} , \lambda_{i+1}^{C} , \dots  , \lambda_n^{C} | \lambda_1^{B} , \dots , \lambda_n^{B} ) = 
S_n (\lambda_1^{C} , \dots , \lambda_n^{C} | \lambda_1^{B} , \dots , \lambda_n^{B} )$ is then valid for $i \in [1, L]$ which allows us to conclude that
\[
\label{sym2}
S_n ( \dots , \lambda_i^{C} , \dots , \lambda_j^{C} , \dots | \lambda_1^{B} , \dots , \lambda_n^{B} ) = S_n ( \dots , \lambda_j^{C} , \dots , \lambda_i^{C} , \dots | \lambda_1^{B} , \dots , \lambda_n^{B}) \; .
\]
It is important to remark here that the symmetry relations (\ref{sym1}) and (\ref{sym2}) have been obtained solely from the
examination of the Eq. (\ref{typeA}). Alternatively, we could also have derived the same symmetry relations from a similar analysis
of the Eq. (\ref{typeD}).

\section{Asymptotic behaviour}
\label{sec:ASYMP}

Considering the variables $x=e^{2 \lambda}$, $q=e^{\gamma}$ and the conventions
described in (\ref{raj}) and (\ref{alfa}), in the limit $x \rightarrow \infty$ we find
\begin{align}
\alpha & \sim \frac{x^{\frac{1}{2}} q^{\frac{1}{2}}}{2} K  & \beta & \sim \frac{q - q^{-1}}{2} X^{-} \nonumber \\
\gamma & \sim \frac{q - q^{-1}}{2} X^{+}  & \delta & \sim \frac{x^{\frac{1}{2}} q^{\frac{1}{2}}}{2} K^{-1} \; ,
\end{align}
where the operators $K$ and $X^{\pm}$ are explicitly given by
\<
K = \left( \begin{matrix}
q^{\frac{1}{2}} & 0 \cr
0 & q^{-\frac{1}{2}} \end{matrix} \right)
\qquad X^{\pm} = \frac{1}{2} \left( \begin{matrix}
0 & 1 \pm 1 \cr
1 \mp 1 & 0 \end{matrix} \right) \; .
\>
In this particular limit the relations (\ref{reABCD}) can be easily iterated
and we obtain 
\<
\label{BC}
B(\lambda) &\sim & \frac{(q - q^{-1})}{2} q^{\frac{L-1}{2}} x^{\frac{L-1}{2}} e^{-\sum_{j=1}^{L} \mu_j} \sum_{k=1}^{L} e^{\mu_k} P_k^{-} \nonumber \\
C(\lambda) &\sim & \frac{(q - q^{-1})}{2} q^{\frac{L-1}{2}} x^{\frac{L-1}{2}} e^{-\sum_{j=1}^{L} \mu_j} \sum_{k=1}^{L} e^{\mu_k} P_k^{+} \; ,
\>
with operators $P_j^{\pm}$ being defined as
\<
P_j^{\pm} = \bigotimes_{k=1}^{j-1} K^{\mp 1} \otimes X^{\pm} \otimes \bigotimes_{k=j+1}^{L} K^{\pm 1} \; .
\>
In their turn the operators $K^{\pm 1}$ and $X^{\pm}$ satisfy the $q$-deformed $\alg{su}(2)$ algebra, 
\<
K X^{\pm} K^{-1} &=& q^{\pm 1} X^{\pm} \nonumber \\
\left[ X^{+} , X^{-}  \right] &=& \frac{K^2 - K^{-2}}{q - q^{-1}} \; ,
\>
which allows us to demonstrate the relations
\<
\label{PP}
P_i^a P_j^b &=& q^{(a,b)} P_j^b P_i^a \qquad \qquad (i < j) \nonumber \\
( P_i^a )^2 &=& 0 
\>
with symbols $(a,b)$ being defined as
\<
(a,b) = \begin{cases} 2 \quad \quad \; a = \pm , b = \pm \cr -2 \quad \;\; a = \pm , b = \mp \end{cases} \; .
\>

Now with the help of the relations (\ref{PP}) and considering $x_i = e^{2 \lambda_i}$, in the limit 
$x_i \rightarrow \infty$ we obtain the expression
\<
\label{Binf}
\prod_{i=1}^{n} B(x_i) & \sim & \frac{(q - q^{-1})^n}{2^{n L}} q^{n\frac{(L-1)}{2} - n(n-1)} [ n! ]_{q^2} e^{-n \sum_{j=1}^{L} \mu_j} \prod_{i=1}^{n} x_i^{\frac{L-1}{2}} \nonumber \\
&& \qquad \qquad \qquad \times \sum_{1 \leq a_1 < \dots < a_n \leq L} e^{\sum_{j=1}^{n} \mu_{a_j}} \mathop{\overrightarrow\prod}\limits_{1 \le j \le n } P_{a_j}^{-} \; ,
\>
and similarly
\<
\label{Cinf}
\prod_{i=1}^{n} C(x_i) & \sim & \frac{(q - q^{-1})^n}{2^{n L}} q^{n\frac{(L-1)}{2} - n(n-1)} [ n! ]_{q^2} e^{-n \sum_{j=1}^{L} \mu_j} \prod_{i=1}^{n} x_i^{\frac{L-1}{2}} \nonumber \\
&& \qquad \qquad \qquad \times \sum_{1 \leq a_1 < \dots < a_n \leq L} e^{\sum_{j=1}^{n} \mu_{a_j}} \mathop{\overrightarrow\prod}\limits_{1 \le j \le n } P_{a_j}^{+} \; . 
\>
The term $[ n! ]_{q^2}$ appearing in (\ref{Binf}) and (\ref{Cinf}) corresponds to the $q$-factorial function and it
is defined as $[ n! ]_{q^2} = 1 (1 + q^2)(1+ q^2 + q^4) \dots (1+ q^2 + \dots + q^{2(n-1)})$.

The relations (\ref{Binf}) and (\ref{Cinf}) can now be combined according to (\ref{scp}) to obtain the 
following formula,
\<
\label{Sinf}
S_n (x_1^C , \dots , x_n^C | x_1^B , \dots , x_n^B ) &\sim & \frac{(q - q^{-1})^{2n}}{2^{2 n L}} q^{n(L-1) - 2 n(n-1)} [ n! ]_{q^2}^2 e^{-2 n \sum_{j=1}^{L} \mu_j} \prod_{i=1}^{n} (x_i^B x_i^C)^{\frac{L-1}{2}}  \nonumber \\
&& \quad  \times \sum_{\stackrel{1 \leq a_1 < \dots < a_n \leq L}{1 \leq b_1 < \dots < b_n \leq L}} e^{\sum_{j=1}^{n} \mu_{a_j} + \mu_{b_j}} \bra{0} \mathop{\overrightarrow\prod}\limits_{1 \le j \le n } P_{a_j}^{+} \mathop{\overrightarrow\prod}\limits_{1 \le j \le n } P_{b_j}^{-} \ket{0} \; , \nonumber \\
\>
in the limit $x_i^{B,C} \rightarrow \infty$. The next step is to compute the quantity $\bra{0} \mathop{\overrightarrow\prod}\limits_{1 \le j \le n } P_{a_j}^{+} \mathop{\overrightarrow\prod}\limits_{1 \le j \le n } P_{b_j}^{-} \ket{0}$
appearing in (\ref{Sinf}), and for that we notice that under the constraint $1 \leq a_1 < \dots < a_n \leq L$ and
$1 \leq b_1 < \dots < b_n \leq L$ the aforementioned quantity will contribute only when $b_j = a_j$. 
In that case we have $\bra{0} \mathop{\overrightarrow\prod}\limits_{1 \le j \le n } P_{a_j}^{+} \mathop{\overrightarrow\prod}\limits_{1 \le j \le n } P_{a_j}^{-} \ket{0} = q^{n(n-1)}$
and the relation (\ref{Sinf}) then simplifies to
\<
\label{SSinf}
S_n  \sim  \frac{(q - q^{-1})^{2n}}{2^{2 n L}} q^{n(L-n)} [ n! ]_{q^2}^2 e^{-2 n \sum_{j=1}^{L} \mu_j}  \sum_{1 \leq a_1 < \dots < a_n \leq L} e^{2 \sum_{j=1}^{n} \mu_{a_j}} \prod_{i=1}^{n} (x_i^B x_i^C)^{\frac{L-1}{2}}  \; . \nonumber \\
\>

\section{Off-shell solution for the case $n=1$}
\label{sec:n1}

The relation (\ref{omij}) offers a way of building the solution of (\ref{typeA}, \ref{typeD})
recursively. More precisely, the formula (\ref{omij}) establishes a relation between the scalar
product $S_n$ and an auxiliary function $V$ consisting essentially of $S_{n-1}$ for a
lattice of length $L-1$. Thus, since we are considering $n \leq L$, at the final step of 
this iteration procedure we shall need the solution of  Eqs. (\ref{typeA})
and (\ref{typeD}) for the case $n=1$. In that case our system of functional equations explicitly
read
\<
\label{AD1}
M_0 S_1 (\lambda_1^C | \lambda_1^B) + N_1^{(B)} S_1 (\lambda_1^C | \lambda_0) + N_1^{(C)} S_1 (\lambda_0 | \lambda_1^B) &=& 0 \nonumber \\
\widetilde{M}_0 S_1 (\lambda_1^C | \lambda_1^B) + \widetilde{N}_1^{(B)} S_1 (\lambda_1^C | \lambda_0) + \widetilde{N}_1^{(C)} S_1 (\lambda_0 | \lambda_1^B) &=& 0
\>
with coefficients
\begin{align}
M_0  = & \frac{c(\lambda_0 - \lambda_1^B)}{b(\lambda_0 - \lambda_1^B)} \frac{b(\lambda_1^B - \lambda_1^C)}{b(\lambda_0 - \lambda_1^C)} \prod_{j=1}^{L} a(\lambda_0 - \mu_j) 
& N_1^{(B)} & =  \frac{c(\lambda_1^B - \lambda_0)}{b(\lambda_1^B - \lambda_0)} \prod_{j=1}^{L} a(\lambda_1^B - \mu_j) \nonumber \\ 
\widetilde{M}_0  = & \frac{c(\lambda_0 - \lambda_1^B)}{b(\lambda_0 - \lambda_1^B)} \frac{b(\lambda_1^C - \lambda_1^B)}{b(\lambda_0 - \lambda_1^C)} \prod_{j=1}^{L} b(\lambda_0 - \mu_j)
& \widetilde{N}_1^{(B)} & =  \frac{c(\lambda_0 - \lambda_1^B)}{b(\lambda_0 - \lambda_1^B)} \prod_{j=1}^{L} b(\lambda_1^B - \mu_j) \nonumber \\
N_1^{(C)} = & \frac{c(\lambda_0 - \lambda_1^C)}{b(\lambda_0 - \lambda_1^C)} \prod_{j=1}^{L} a(\lambda_1^C - \mu_j) & \widetilde{N}_1^{(C)} & = \frac{c(\lambda_1^C - \lambda_0)}{b(\lambda_1^C - \lambda_0)} \prod_{j=1}^{L} b(\lambda_1^C - \mu_j) \; , \nonumber \\
\end{align}
and the solution of (\ref{AD1}) can be obtained as follows. By eliminating the term $S_1 (\lambda_0 | \lambda_1^B)$
from (\ref{AD1}), we obtain an equation involving only the terms $S_1 (\lambda_1^C | \lambda_1^B)$
and $S_1 (\lambda_1^C | \lambda_0)$. More precisely, we are left with the identity
\<
\label{SS}
\frac{b(\lambda_1^B - \lambda_1^C) S_1 (\lambda_1^C | \lambda_1^B)} {\left[ \prod_{j=1}^{L} \frac{a(\lambda_1^B - \mu_j)}{a(\lambda_1^C - \mu_j)}  - \prod_{j=1}^{L} \frac{b(\lambda_1^B - \mu_j)}{b(\lambda_1^C - \mu_j)} \right]} = \frac{b(\lambda_0 - \lambda_1^C) S_1 (\lambda_1^C | \lambda_0)} {\left[ \prod_{j=1}^{L} \frac{a(\lambda_0 - \mu_j)}{a(\lambda_1^C - \mu_j)}  - \prod_{j=1}^{L} \frac{b(\lambda_0 - \mu_j)}{b(\lambda_1^C - \mu_j)} \right]} \; ,
\>
where the variable $\lambda_1^C$ now assumes the role of a parameter. The relation (\ref{SS}) is then readily solved by
\<
\label{SS1}
S_1 (\lambda_1^C | \lambda_1^B ) = \frac{F(\lambda_1^C)}{b(\lambda_1^C - \lambda_1^B)}
\left[ \prod_{j=1}^{L} \frac{a(\lambda_1^B - \mu_j)}{a(\lambda_1^C - \mu_j)} - \prod_{j=1}^{L} \frac{b(\lambda_1^B - \mu_j)}{b(\lambda_1^C - \mu_j)} \right] \; ,
\>
where $F$ is an arbitrary function. Next we substitute the expression (\ref{SS1}) into the first Eq. 
of (\ref{AD1}), and after eliminating an overall factor we are left with the relation
\<
\label{SS2}
&& F(\lambda_1^C) \left[ \prod_{j=1}^{L} \frac{a(\lambda_1^B - \mu_j) b(\lambda_0 - \mu_j)}{b(\lambda_1^C - \mu_j)} - \prod_{j=1}^{L} \frac{a(\lambda_0 - \mu_j) b(\lambda_1^B - \mu_j)}{b(\lambda_1^C - \mu_j)} \right] \nonumber \\
&& - F(\lambda_0) \left[ \prod_{j=1}^{L} \frac{a(\lambda_1^B - \mu_j) a(\lambda_1^C - \mu_j)}{a(\lambda_0 - \mu_j)} - \prod_{j=1}^{L} \frac{a(\lambda_1^C - \mu_j) b(\lambda_1^B - \mu_j)}{b(\lambda_0 - \mu_j)} \right] = 0 \; .
\>
As far as the Eq. (\ref{SS2}) is concerned, the variable $\lambda_1^B$ is merely a parameter which can be set
at our convenience. In particular, for the specialisation $\lambda_1^B = \mu_k - \gamma$ the Eq. (\ref{SS2})
simplifies to 
\<
\label{SS3}
\frac{F(\lambda_0)}{\prod_{j=1}^{L} a(\lambda_0 - \mu_j) b(\lambda_0 - \mu_j)} = \frac{F(\lambda_1^C)}{\prod_{j=1}^{L} a(\lambda_1^C - \mu_j) b(\lambda_1^C - \mu_j)} \; .
\>
Eq. (\ref{SS3}) can then be easily solved and we find
\<
\label{ff}
F(\lambda) = \Omega \prod_{j=1}^{L} a(\lambda - \mu_j) b(\lambda - \mu_j)
\>
where $\Omega$ is a $\lambda$ independent parameter. Then the combination of (\ref{SS1}) and (\ref{ff}) yields the 
expression
\<
\label{sol1}
S_1 (\lambda_1^C | \lambda_1^B) = \frac{c(\lambda_1^C - \lambda_1^B)}{b(\lambda_1^C - \lambda_1^B)} \left[ \prod_{j=1}^{L} a(\lambda_1^B - \mu_j) b(\lambda_1^C - \mu_j) - \prod_{j=1}^{L} a(\lambda_1^C - \mu_j) b(\lambda_1^B - \mu_j) \right] \; , \nonumber \\
\>
where the constant $\Omega$ has been fixed by the asymptotic behaviour (\ref{SSinf}).

\bibliographystyle{hunsrt}
\bibliography{references}

\end{document}